\documentclass[aap,preprint]{imsart}
\usepackage{amsthm,amsmath,amsfonts,amssymb}
\usepackage{color}
\usepackage{tikz,pgfplots}

\startlocaldefs
\newcommand{\dd}{\mathrm{d}}
\newcommand{\bbR}{\mathbb{R}}
\newcommand{\normal}{\mathrm{N}}
\newcommand{\tsfrac}[2]{{\textstyle \frac{#1}{#2}}}
\newcommand{\twobyone}[2]{\begin{array}{c} #1 \\ #2 \end{array}}
\newcommand{\twobytwo}[4]{\begin{array}{cc} #1 & #2 \\ #3 & #4 \end{array}}
\newcommand{\nrm}[1]{\Vert #1 \Vert}
\newcommand{\iid}{\stackrel{\text{i.i.d.}}{\sim}}

\newtheorem{theorem}{Theorem}[section]
\newtheorem{lemma}[theorem]{Lemma}
\newtheorem{corollary}[theorem]{Corollary}
\newtheorem{assumption}[theorem]{Assumption}
\endlocaldefs

\begin{document}

\begin{frontmatter}

\title{Tuning of MCMC with Langevin, Hamiltonian, and other stochastic autoregressive proposals}
\runtitle{MCMC with autoregressive proposals}

\begin{aug}
\author{\fnms{Richard A.} \snm{Norton}\corref{}\ead[label=e1]{richard.norton@otago.ac.nz}}
\and
\author{\fnms{Colin} \snm{Fox}\ead[label=e2]{colin.fox@otago.ac.nz}}

\runauthor{R. A. Norton and C. Fox}

\affiliation{University of Otago}
\address{Richard Norton\\Department of Mathematics and Statistics\\University of Otago\\PO Box 56\\Dunedin 9054\\ New Zealand\\ \printead{e1}}

\address{Colin Fox\\Department of Physics\\University of Otago\\PO Box 56\\Dunedin 9054\\ New Zealand\\ \printead{e2}}
\end{aug}

\begin{abstract}
Proposals for Metropolis-Hastings MCMC derived by discretizing Langevin diffusion or Hamiltonian dynamics are examples of stochastic autoregressive proposals
that form a natural wider class of proposals with equivalent computability. 
We analyze Metropolis-Hastings MCMC with stochastic autoregressive proposals applied to target distributions that are absolutely continuous with respect to some Gaussian distribution to derive expressions for expected acceptance probability and expected jump size, as well as measures of computational cost, in the limit of high dimension.
Thus, we are able to unify existing analyzes for these classes of proposals, and to extend the theoretical results that provide useful guidelines for tuning the proposals for optimal computational efficiency.
For the simplified Langevin algorithm we find that it is optimal to take at least three steps of the proposal before
the Metropolis-Hastings accept-reject step, and for Hamiltonian/hybrid Monte Carlo we provide new guidelines for the optimal number of integration steps and
criteria for choosing the optimal mass matrix.
\end{abstract}



\end{frontmatter}

\section{Introduction}
\label{sec intro}

We consider Metropolis-Hastings (MH) algorithms for sampling from a target distribution $\pi_d$ using a first-order stochastic autoregressive (AR(1)) process proposal with Gaussian `noise'; given current state $x \in \bbR^d$ the proposal $y \in \bbR^d$ is given by
\begin{equation}
\label{eq:ar1}
	y = G x + g + \nu
\end{equation}
where $G \in \bbR^{d\times d}$ is the iteration matrix, $g \in \bbR^d$ is a fixed vector, and $\nu\stackrel{\text{i.i.d.}}{\sim}\normal(0,\Sigma)$ where $\Sigma \in \bbR^{d \times d}$ is symmetric and positive definite (s.p.d.).  We will refer to \eqref{eq:ar1} as a  \emph{stochastic AR(1) proposal}.  One iteration of a chain with MH dynamics is simulated by accepting the proposal $y$ with probability
\begin{equation}
\label{eq alpha}
	\alpha(x,y) = 1 \wedge \frac{\pi_d(y) q(y,x)}{\pi_d(x) q(x,y)}
\end{equation}
and otherwise remaining at $x$. Here $\pi_d(\cdot)$ denotes the target probability density function, $q(x,\dd y) = q(x,y) \dd y$ is the transition kernel for the proposal $y$ given current state $x$, and $a \wedge b = \min\{a,b\}$.  

We consider \emph{stationary} stochastic AR(1) proposals~\eqref{eq:ar1} for which $G$, $g$ and $\Sigma$ are fixed and the \emph{proposal chain}, generated by repeated iterates of~\eqref{eq:ar1}, is \emph{convergent} in distribution, which occurs iff $\rho(G)<1$ where $\rho(G)$ denotes the spectral radius of $G$~\cite{Duflo, DiaconisFreedman, FP2016}. Then this proposal chain converges to the \emph{proposal equilibrium distribution} $\normal(\mathcal{A}^{-1}\beta,\mathcal{A}^{-1})$~\cite{FP2016}; see Theorem~\ref{thm:2.1}.    Such chains are precisely the convergent (generalized)  Gibbs samplers for normal targets~\cite{FP2016} including blocking and reparametrization (\emph{preconditioning} in numerical analysis), with the equivalence to stationary linear iterative solvers affording extensive detail about the chain, such as the $n$-step distribution, error polynomial, and convergence rate, as well as acceleration~\cite{FP2014,FP2016}.

Stationary and non-stationary (where $G$, $g$ and $\Sigma$ may depend on $x$) stochastic AR(1) proposals occur in many popular methods including those using discretized Langevin diffusion such as the Metropolis-adjusted Langevin algorithm (MALA)~\cite{RR1998}, the simplified Langevin algorithm (SLA), the $\theta$-SLA method and preconditioned versions of SLA~\cite{BRS2009}, the Crank-Nicolson (CN), preconditioned Crank-Nicolson (pCN), and preconditioned Crank-Nicolson Langevin (pCNL)~\cite{CRSW2013} algorithms, and discretized Hamiltonian dynamics used in  hybrid Monte Carlo (HMC)~ \cite{DKPR1987,BPRSS2010,N1993}, as well as the stochastic Newton and scaled stochastic Newton algorithms~\cite{MWBG12,BG2014}. Our restriction to convergent proposal chains precludes the random-walk Metropolis algorithm (RWM) \cite{RGG1997} since then $G=I$.

These, and other, existing MH MCMC algorithms derive the proposal by discretizing a stochastic differential equation such as Langevin diffusion \cite{RGG1997,RR1998,RR2001,B2007,B2008,BR2008,BRS2009}, or a randomized Hamiltonian dynamical system~\cite{BPRSS2010,BPSSS2011}, for which the continuous process targets $\pi_d$. This leads to the limited range of stochastic AR(1) proposals that have been considered. 

We prefer working directly with the AR(1) process rather than discretizations of differential equations for several reasons: stochastic AR(1) proposals generalize both Langevin and Hamiltonian proposals; this discrete process is natural on digital computers with the computational cost being evident from the form of $G$; subsequent analysis does not require bringing in notions from differential equations that are not intrinsic to the sampling problem at hand. 

Since a convergent stochastic AR(1) process converges to a normal distribution, see Theorem \ref{thm:2.1}, it follows that a MH algorithm with a stochastic AR(1) proposal is efficient when its proposal chain converges rapidly to a normal approximation of $\pi_d$.  Hence all these algorithms may be viewed as taking one or more steps of a Gibbs sampler that targets normal approximations to $\pi_d$, followed by the MH accept-reject step.  Stationary proposals target a single global normal approximation to $\pi_d$, while non-stationary proposals target local normal approximations to $\pi_d$.  Quality of the resulting MH algorithm, in the sense of mixing and convergence rates, depends on the quality of these normal approximations and the rate of convergence of the proposal chain; we formalize this idea for stationary proposals in Section~\ref{sec nongaussian}, though our results also describe the local behavior of algorithms with non-stationary proposals.

We consider target distributions that are a change of measure from a Gaussian reference measure, i.e.
\begin{equation}
\label{eq target0}
	\frac{\dd \pi_d}{\dd \tilde{\pi}_d}(x) = \exp( - \phi_d(x))
\end{equation}
for some $\phi_d: \bbR^d \mapsto \bbR$ where $\tilde{\pi}_d$ is a Gaussian, i.e.,
\begin{equation}
\label{eq ref}
	\tilde{\pi}_d(x) \propto \exp \left( - \frac{1}{2} x^T A x + b^T x \right)
\end{equation}
for s.p.d. $A \in \bbR^{d \times d}$ and $b \in \bbR^d$.  

This setting arises, for instance, in Bayesian formulations of inverse problems when a Gaussian smoothness prior is used and $-\phi_d$ is the log likelihood, see e.g.~\cite{S2010,HSV2011,BS2009b}. We will require the same conditions on $\phi_d$ as in \cite{BRS2009}, essentially that $\phi_d$ is either bounded, or is bounded below, satisfies a type of Lipschitz continuity, and has bounded growth; see Section~\ref{sec nongaussian}.  Note that $A$ is not required to be diagonal in our analysis or for computation, and hence we are able to extend the range of applicability of results derived in~\cite{BRS2009,BPRSS2010,BPSSS2011,BS2009a,BS2009b} since there is no need to compute a spectral decomposition of $A$, which is typically infeasible in high dimensions.

The case where $\tilde{\pi}_d$ has product form
\begin{equation}
\label{eq target1}
	\tilde{\pi}_d(x) = \prod_{i=1}^d \lambda_i f\left( \lambda_i x_i  \right)
\end{equation}
with $\log f$ quadratic and $\{ \lambda_i \}_{i=1}^d \subset (0,\infty)$ is a special case of our theory, but we do not consider non-quadratic $\log f$.  

Those algorithms that have been shown to have dimension-independent mixing~\cite{BRSV2008,BPSSS2011,CRSW2013,ABPS2014}, defined by a strictly positive expected jump as $d\rightarrow \infty$, correspond to stationary proposals such that the proposal equilibrium distribution $\normal(\mathcal{A}^{-1}\beta,\mathcal{A}^{-1})$ equals the Gaussian reference measure $\normal(A^{-1}b,A^{-1})$.  We show that this is not necessary to achieve dimension-independent mixing, see Corollary \ref{cor:dimindep}, but for an efficient MCMC near equality between $\normal(\mathcal{A}^{-1}\beta,\mathcal{A}^{-1})$ and $\normal(A^{-1}b,A^{-1})$ is very desirable. In particular, we extend the conditions that guarantee dimension-independent sampling, and enable more efficient sampling.
The proposals for MALA, SLA and HMC satisfy $\mathcal{A}^{-1}\beta = A^{-1}b$, but not $\mathcal{A} = A$.  Since $\mathcal{A}$ is not sufficiently close to $A$ in these methods, they are not dimension independent.

All convergent stochastic AR(1) proposals with equilibrium distribution $\normal(\mathcal{A}^{-1}\beta,\mathcal{A}^{-1})$, including the case when $\normal(\mathcal{A}^{-1}\beta,\mathcal{A}^{-1})$ and $\normal(A^{-1}b,A^{-1})$ are identical, may be found using \emph{matrix splitting} of $\mathcal{A}$ to find $G$, $g$ and $\Sigma$, see \cite{F2013,FP2014,FP2016}, which also gives rates of convergence, etc. In this sense, we include the analysis of both Langevin diffusion and Hamiltonian dynamics based MCMC methods within our unified analysis of MH MCMC algorithms with convergent stochastic AR(1) proposals.

For example, when the target is $\normal(A^{-1}b,A^{-1})$, then MALA and SLA are the same, having stationary stochastic AR(1) proposal \eqref{eq:ar1} with $G = I-\frac{h}{2}A$, $g = \frac{h}{2} b$ and $\Sigma = hI$ for some $h >0$, and the proposal equilibrium is $\normal(\mathcal{A}^{-1}\beta,\mathcal{A}^{-1})$ with $\mathcal{A}^{-1} \beta = A^{-1}b$ and $\mathcal{A} = A - \frac{h}{2}A^2 \neq A$, see Theorem \ref{thm:2.1}.  The Unadjusted Langevin Algorithm (ULA) corresponds to MALA's proposal chain \cite{RT1996}, so we have also characterized ULA's incorrect equilibrium distribution and convergence rate when the target is normal.

For Langevin diffusion-based proposals, we identify the proposal equilibrium distribution, recover and extend the existing theory by allowing $\tilde{\pi}_d$ to be a non-product distribution, and we quantify the effect of `preconditioning' the Langevin diffusion.

HMC is also an example of an algorithm with a stochastic AR(1) proposal.  For normal targets, our theory extends the available analyses of HMC to non-product target distributions, we characterize the spectrum of the iteration matrix $G$, determine the optimal integration time, and characterize the effect of a general mass matrix.  Existing analyses of HMC only tell us to tune the step size until the observed acceptance probability is $0.651$, for the case when the target distribution has i.i.d. product form \cite{BPRSS2010}.

Since we are not limited to discretizing a differential equation to construct a proposal, we can analyze more general algorithms such as multi-step proposals that take $L$ steps of a stationary stocastic AR(1) proposal before The MH accept/reject step.  For multi-step SLA, it is more efficient to take $L \geq 3$ than $L=1$ when the dominant computational costs are multiplying by $A$ and evaluating $\phi_d$, see Section \ref{sec lstep}.  Multi-step SLA should be tuned so that the acceptance probability is $0.574$, the same as $1$-step SLA.

Our analysis is for the case where $G$ and $\Sigma$ are functions of $A$, which allows a diagonalizing coordinate transformation.  This restriction is motivated by practicalities of high-dimensional computation, particularly high-dimensional inverse problems, where the palette of feasible computation is limited to operation by $A$ (and $A^T$, though $A$ is symmetric in our case) and evaluating $\phi_d$; functions of $A$ may be evaluated as rational approximations (see e.g.~\cite{Higham}).  Not all existing MCMC algorithms that we consider satisfy this computational feasibility criterion, notably pCN and pCNL, and the version of HMC in \cite{BPSSS2011}, however a change of coordinates is allowed for analysis.

We derive expressions for the expected acceptance probability and the expected jump size of the Markov chain.  By using the expected jump size as a proxy for statistical efficiency, and considering the cost of computing \eqref{eq:ar1} and \eqref{eq alpha} we can determine the computational efficiency of an algorithm.

The statistical efficiency of a MCMC method is often measured by the integrated autocorrelation time for a statistic of concern.  We are unable to directly estimate this quantity.  Instead, we calculate the expected squared jump size of the Markov chain in a direction $q \in \bbR^d$,
$$
	\mathrm{E}[(q^T(x' - x))^2],
$$
where $x$ and $x'$ are successive states of the Markov chain in equilibrium,
that is related to the integrated autocorrelation time for the linear functional $q^T(\cdot)$ by
$$
	\mathrm{Corr}[q^T x,q^T x'] = 1 - \frac{\mathrm{E}[(q^T(x' - x))^2]}{2 \mathrm{Var}[q^T x]}.
$$
In particular, large jump size implies small first-order autocorrelation, see e.g. \cite[\S 3]{RR1998} or \cite[\S 2.3]{BRS2009}.  This approach is similar to that used in~\cite{BPRSS2010,BRS2009} for analyzing the efficiency of RWM, MALA and HMC where the expected squared jump size of a single component of the state of the Markov chain is considered.  

This measure of statistical efficiency does not always depend on $q$.  For high-dimensional problems where the MH algorithm converges to a diffusion process, then all $q$ essentially lead to the same notion of efficiency \cite[\S 2.2]{RR2001}.  We observe this in Theorems \ref{thm genlang conv} and \ref{thm:lstep} where jump size is independent of $q$.  However, in HMC this definition of jump size is dependent on $q$, see Theorem~\ref{thm hmc conv}.

The remainder of this article is organized as follows.  We identify the equilibrium distribution for a given stochastic AR(1) process in Section \ref{sec stationary}.  By exploiting the fact that $G$ and $\Sigma$ are functions of $A$, we use a coordinate transformation to diagonalize the proposal and target in Section \ref{sec diagonalization}.  Sections \ref{sec gaussian} and \ref{sec nongaussian} provide the main body of theory where we prove results in the Gaussian and non-Gaussian target cases, respectively.  We apply our theory to examples in Section~\ref{sec examples}, and we conclude with some discussion in Section \ref{sec discussion}.  Most proofs are in the appendix.  

\subsection{Notation}

Let $G_i$, $\lambda_i^2$ and $\tilde{\lambda}_i^2$ be eigenvalues of $G$, $A$ and $\mathcal{A}$, respectively. Let $\mu_i = (A^{-1}b)_i$ and $\tilde{\mu}_i = (\mathcal{A}^{-1}\beta)_i$ be components of the means of $\tilde{\pi}_d$ and $\normal(\mathcal{A}^{-1}\beta,\mathcal{A}^{-1})$, respectively.  Define
\begin{align*}
	\tilde{g}_i &:= 1-G_i, & \hat{g}_i &:= 1-G_i^2, & 
	r_i &:= \frac{\lambda_i^2-\tilde{\lambda}_i^2}{\lambda_i^2}, &
	\tilde{r}_i &:= \frac{\lambda_i^2}{\tilde{\lambda}_i^2}, & 
	\hat{r}_i &:= \mu_i - \tilde{\mu}_i
\end{align*}
so that $\tilde{g}_i$ and $\hat{g}_i$ quantify the gap between the spectrum of $G$ and $1$, and $r_i$, $\tilde{r}_i$ and $\hat{r}_i$ quantify the difference between covariances and means of $\tilde{\pi}_d$ and $\normal(\mathcal{A}^{-1}\beta,\mathcal{A}^{-1})$, respectively.  To avoid lengthy expressions also define
\begin{align*}
	T_{0i} &:= \hat{r}_i^2 \lambda_i^2 ( \tsfrac{1}{2} r_i \hat{g}_i - \tilde{g}_i), &
	T_{1i} &:= \hat{r}_i \lambda_i ( r_i \hat{g}_i - \tilde{g}_i), \\
	T_{2i} &:= \hat{r}_i \lambda_i (\tilde{r}_i \hat{g}_i )^{1/2} ( 1 - r_i G_i), &
	T_{3i} &:= \tsfrac{1}{2} r_i \hat{g}_i, \\
	T_{4i} &:= -\tsfrac{1}{2} r_i \tilde{r}_i \hat{g}_i, &
	T_{5i} &:= -r_i G_i (\tilde{r}_i \hat{g}_i)^{1/2}.
\end{align*}
In general, all of these quantities may depend on $d$.  The standard normal cumulative distribution function will always be $\Phi$.

We say $f_{d,i} = \mathcal{O}(g_{d,i})$ (uniformly in $i$) as $d \rightarrow \infty$ if for all $i$ and all sufficiently large $d$, $f_{d,i}/g_{d,i}$ is bounded by a constant that is independent of $d$ and $i$. Likewise, $f_{d,i} = \mathrm{o}(g_{d,i})$ (uniformly in $i$) as $d \rightarrow \infty$ if $\max_{1\leq i \leq d} f_{d,i}/g_{d,i} \rightarrow 0$ as $d \rightarrow \infty$.  For brevity we sometimes omit ``uniformly in $i$''.

Other articles use $\lambda_i^2$ to denote the eigenvalues of the covariance matrix corresponding to $\tilde{\pi}_d$ \cite{BRS2009,BS2009b,BS2009a,BPSSS2011}.  We do not follow this convention and instead use $\lambda_i^2$ to denote eigenvalues of the precision matrix ($\lambda_i$ is the inverse of a standard deviation).  Since sampling from $N(A^{-1}b,A^{-1})$ is closely related to solving $Ax=b$ (see \cite{F2013,FP2014,FP2016}), our notation aligns with literature on solving linear systems.

\section{Stationary distribution for an AR(1) process}
\label{sec stationary}

The following theorem determines the equilibrium distribution for a convergent stochastic AR(1) process.

\begin{theorem}
\label{thm:2.1}
If $\rho(G) < 1$, then the stochastic AR(1) process defined by \eqref{eq:ar1} (the proposal chain) converges to $\normal(\mathcal{A}^{-1}\beta,\mathcal{A}^{-1})$ where
\begin{equation}
\label{eq calA}
	\mathcal{A} = \left( \sum_{l=0}^\infty G^l \Sigma (G^{T})^l \right)^{-1} \quad \mbox{is s.p.d. and} \quad
	\beta = \mathcal{A} (I-G)^{-1} g.
\end{equation}
\end{theorem}

\begin{proof}
Define $\mathcal{A}^{-1} := \sum_{l=0}^\infty G^l \Sigma (G^{T})^l$.  Since $\rho(G)<1$ and $\Sigma$ is s.p.d., $\mathcal{A}^{-1}$ is well-defined and s.p.d..  The $n^{\mathrm{th}}$ step of the AR(1) process, starting at $x^{(0)}$, where $\xi^{(i)} \stackrel{\text{i.i.d.}}{\sim} \normal(0,I)$, is 
$$
	x^{(n)} = G^n x^{(0)} + \sum_{i=1}^n G^{i-1} g + \sum_{i=1}^n G^{i-1} \Sigma^{1/2} \xi^{(i)}.
$$
Since $\rho(G)<1$, it follows that $x^{n}$ converges to the stationary distribution $N((I-G)^{-1}g, \mathcal{A}^{-1})$ as $n \rightarrow \infty$.
\end{proof}

\begin{corollary}
$\Sigma = \mathcal{A}^{-1} - G \mathcal{A}^{-1} G^T$.
\end{corollary}

If $G$ and $\Sigma$ are functions of $A$ then $G \Sigma$ is symmetric and the following corollary and lemma apply.

\begin{corollary}
\label{lem equiv}
If $G \Sigma$ is symmetric, then $\Sigma^{-1} G$ is also symmetric, and $\mathcal{A} = \Sigma^{-1} (I-G^2) = \Sigma^{-1} - G^T \Sigma^{-1} G$.
\end{corollary}

\begin{lemma}
\label{lem reversible}
If $G \Sigma$ is symmetric, then the stochastic AR(1) process defined by \eqref{eq:ar1} is $\normal(\mathcal{A}^{-1}\beta,\mathcal{A}^{-1})$-reversible.
\end{lemma}

\begin{proof}
Use identities $\Sigma = \mathcal{A}^{-1} - G \mathcal{A}^{-1} G^T$, $\mathcal{A} = \Sigma^{-1} (I-G^2) = \Sigma^{-1} - G^T \Sigma^{-1} G$, and $\beta = \mathcal{A} (I-G)^{-1} g$ to check detailed balance, see \cite[p.113]{Liu2001book}, i.e. $\pi^*(x) q(x,y) = \pi^*(y) q(y,x)$ for all $x,y \in \mathbb{R}^d$, where $\pi^*(x) \propto \exp( - \frac{1}{2} x^T \mathcal{A} x + \beta^T x)$ and $q(x,y) \propto \exp( -\frac{1}{2} (y-Gx-g)^T \Sigma^{-1} (y-Gx-g))$.
\end{proof}

\section{Diagonalization of the proposal and target}
\label{sec diagonalization}

The results in \cite{BRS2009} for SLA and RWM, where $\tilde{\pi}_d$ has product form \eqref{eq target1}, easily extend to the case where $\tilde{\pi}_d$ is Gaussian \eqref{eq ref} with $A$ that may have non-zero off-diagonal terms.  This is obvious once we recognize that the MH transition kernel commutes with orthogonal transformations, and there exists an orthogonal transformation that simultaneously diagonalizes the covariance matrix of $\tilde{\pi}_d$, $G$, and $\Sigma$ in \eqref{eq:ar1}.  

\begin{lemma}
\label{lem:orthog}
Suppose we have a MH algorithm with proposal $q(x,\dd y) = q(x,y) \dd y$ and target $\pi(x)$.  This induces an acceptance probability $\alpha(x,y)$ and transition kernel $P(x,\dd y)$.

Let $W:\bbR^d \rightarrow \bbR^d$ be an invertible homogeneous linear transformation (equivalently, let $W \in \mathbb{R}^{d\times d}$ be an invertible matrix).  The coordinate transformation
$$
	\hat{x} = W^{-1} x \qquad \forall x \in \bbR^d,
$$
induces a new MH algorithm with proposal $\hat{q}(x,\dd y) = \hat{q}(x,y) \dd y$ and target $\hat{\pi}(x)$, which in turn induces an acceptance probability $\hat{\alpha}(x,y)$ and transition kernel $\hat{P}(x,\dd y)$.  

Then 
$$
	\alpha(x,y) = \hat{\alpha}(\hat{x},\hat{y}) \qquad \forall x,y \in \bbR^d,
$$
and
$$
	P^n = W \hat{P}^n W^{-1} \qquad \mbox{for $n=0,1,2,\dotsc$.}
$$
\end{lemma}

\begin{proof}
Let $|\cdot|$ denote determinant.  Then $q(x,y)\dd y = q(W\hat{x},W\hat{y}) |W| \dd \hat{y}$ and $\pi(x)\dd x = \pi(W\hat{x}) |W| \dd \hat{x}$, so $\hat{q}(\hat{x},\hat{y}) = q(W\hat{x},W\hat{y})|W|$ and $\hat{\pi}(\hat{x}) = \pi(W\hat{x}) |W|$.  Using these identities it is easy to show that $\alpha(x,y) = \hat{\alpha}(\hat{x},\hat{y})$, and $P(x,W(B)) = \hat{P}(\hat{x},B)$ for every $x \in \mathbb{R}^d$ and $B \in \mathcal{B}(\mathbb{R}^d)$.  The result follows.  
\end{proof}

A consequence of this result is that the Markov chain corresponding to $P$ and the transformed Markov chain corresponding to $\hat{P}$ have identical convergence properties, so it is sufficient to analyze the properties of $\hat{P}$ to determine the properties of $P$.  
We apply Lemma \ref{lem:orthog} to the following two MH algorithms related by an orthogonal transformation.

Since $A$ is s.p.d. we can define a spectral decomposition
$$
	A = Q \Lambda Q^T
$$ 
where $Q \in \bbR^{d\times d}$ is an orthogonal matrix (orthonormal columns) and $\Lambda = \operatorname{diag}(\lambda_1^2,\dotsc,\lambda_d^2)$ is the diagonal matrix of eigenvalues of $A$.

\begin{lemma}
\label{lem2}
Suppose $G = G(A)$ and $\Sigma = \Sigma(A)$ are functions of $A$.  The coordinate transformation
$$
	\hat{x} = Q^T x
$$
transforms the MH algorithm defined by
\begin{equation}
\label{MH0}
\begin{split}
	\mbox{Target:} & \qquad \normal(A^{-1}b, A^{-1}), \\
	\mbox{Proposal:} & \qquad y = G x + g + \nu, \qquad \nu \iid \normal(0,\Sigma) 
\end{split}
\end{equation}
to a MH algorithm defined by
\begin{equation}
\label{MH1}
\begin{split}
	\mbox{Target:} & \qquad \normal(\Lambda^{-1} Q^T b, \Lambda^{-1}), \\
	\mbox{Proposal:} & \qquad y = G(\Lambda) x + Q^T g + \nu, \qquad \nu \iid \normal(0,\Sigma(\Lambda)). 
\end{split}
\end{equation}
\end{lemma}

Note that MH algorithm \eqref{MH1} has diagonal covariance in the target \emph{and} the random term in the proposal, \emph{and} $G$ is diagonal, so it is much easier to analyze than \eqref{MH0}.  Lemma \ref{lem:orthog} also tells us that it is sufficient to analyze \eqref{MH1} to determine the properties of \eqref{MH0}.  In particular, the expected acceptance probability in equilibrium is identical for \eqref{MH0} and \eqref{MH1}, and the integrated autocorrelation time (in the $2$-norm) is also identical.

Expected jump size in a particular coordinate direction is not preserved under the orthogonal transformation.  Nevertheless, we derive results of this kind for \eqref{MH0} based on the equivalent result for \eqref{MH1}.
While some properties related to statistical efficiency of \eqref{MH0} and \eqref{MH1} are preserved under the orthogonal transformation, the computational cost of the algorithms is \emph{not} preserved.  In particular, the cost of constructing $Q$ can be huge for large $d$, making computations with \eqref{MH1} infeasible.  We propose always computing with \eqref{MH0}, but analyzing \eqref{MH1}, so we never need to compute the orthogonal transformation.  It is sufficient to know that it exists.

In the case when $\phi_d \neq 0$ and the target is a change of measure from a Gaussian, then we need to check that any assumptions on $\phi_d$ still hold under the coordinate transformation.  In the case of \cite{BRS2009}, this is indeed the case.

Thus, we have reduced the study of MH algorithms with AR(1) proposals where $G$ and $\Sigma$ are functions of $A$, targeting distributions that are either Gaussian or a change of measure from a reference Gaussian, to the special case when $G$ and $\Sigma$ are diagonal matrices and the reference Gaussian has a diagonal covariance matrix.

\section{Gaussian targets}
\label{sec gaussian}
\subsection{Expected acceptance probability for a Gaussian target}
\label{sec expect}

The expected acceptance probability is related to efficiency. For optimal performance, the proposal is often tuned so that the observed average acceptance probability is a predetermined value between $0$ and $1$.  For example, in certain settings and as $d \rightarrow \infty$, $0.234$ is optimal for RWM \cite{RGG1997}, $0.574$ is optimal for MALA \cite{RR1998} and SLA \cite{BRS2009}, and $0.651$ is optimal for HMC \cite{BPRSS2010}.  Establishing these results required an expression for the expected acceptance probability of the algorithm as $d \rightarrow \infty$.  We now derive an expression for the expected acceptance probability of Algorithm~\eqref{MH0}, that targets a Gaussian.

\begin{lemma}
\label{lem1}
Suppose $G\Sigma$ is symmetric (holds when $G$ and $\Sigma$ are functions of $A$), then the acceptance probability for \eqref{MH0} satisfies
$$
	\alpha(x,y) = 1 \wedge \exp \left( -\frac{1}{2} y^T (A - \mathcal{A}) y + \frac{1}{2} x^T(A-\mathcal{A}) x + (b-\beta)^T (y-x)\right).
$$
\end{lemma}

\begin{proof}
The result follows by substituting $q(x,y) \propto \exp( -\frac{1}{2} (y-Gx-g)^T \Sigma^{-1} (y-Gx-g))$ and \eqref{eq ref} into \eqref{eq alpha}, using Theorem \ref{thm:2.1} and Corollary \ref{lem equiv}, which imply $\beta^T = g^T (I+G^T) \Sigma^{-1}$. 
\end{proof}

\begin{theorem}
\label{thm accept}
Suppose that $G$ and $\Sigma$ are functions of $A$, and the Markov chain induced by \eqref{MH0} is in equilibrium, i.e. $x \sim \normal(A^{-1}b,A^{-1})$.  
If there exists a $\delta > 0$ such that
\begin{equation}
\label{eq Tcond}
	\lim_{d \rightarrow \infty} \frac{\sum_{i=1}^d |T_{ji}|^{2+\delta}}{ \left( \sum_{i=1}^d |T_{ji}|^2 \right)^{1+\delta/2}} = 0, \qquad j=1,2,3,4,5
\end{equation}
($j=0$ is not required) and the limits $\mu = \lim_{d\rightarrow \infty} \sum_{i=1}^d \mu_{d,i}$ and $\sigma^2 = \lim_{d\rightarrow \infty} \sum_{i=1}^d \sigma_{d,i}^2$ exist, where
$$
	\mu_{d,i} = T_{0i} + T_{3i} + T_{4i} \qquad \mbox{and} \qquad \sigma_{d,i}^2 = T_{1i}^2 + T_{2i}^2 + 2T_{3i}^2 + 2T_{4i}^2 + T_{5i}^2,
$$
then 
$$
	Z:= \log \left( \frac{\pi(y)q(y,x)}{\pi(x)q(x,y)} \right) \xrightarrow{\mathcal{D}} \normal(\mu,\sigma^2) \qquad \mbox{as $d \rightarrow \infty$}
$$
and the expected acceptance probability of \eqref{MH0} satisfies
$$	
	\mathrm{E}[\alpha(x,y)] = \mathrm{E}[1 \wedge \mathrm{e}^Z] \rightarrow \Phi(\tsfrac{\mu}{\sigma}) + \mathrm{e}^{\mu + \sigma^2/2} \Phi(-\sigma - \tsfrac{\mu}{\sigma})
	\qquad \mbox{as $d \rightarrow \infty$.}
$$
\end{theorem}

The proof of Theorem~\ref{thm accept} is in the appendix.  Recall that quantities $\lambda_i$, $\mu_i$, $G_i$, $\tilde{\lambda}_i$, $\tilde{\mu}_i$, $r_i$, $\tilde{r}_i$, $\hat{r}_i$, and hence $T_{ji}$, $\mu_{d,i}$ and $\sigma_{d,i}$, may depend on $d$.

A strictly positive expected acceptance probability is a necessary (but not sufficient) condition for an efficient proposal (except in pathological cases).

\begin{corollary}
\label{cor:accept}
In Theorem \ref{thm accept}, if $\sigma$ is finite and $\mu > -\infty$ then $\lim_{d\rightarrow\infty} \mathrm{E}[\alpha(x,y)] > 0$.
\end{corollary}

\begin{lemma}
\label{lem accept}
With the same definitions as Theorem \ref{thm accept}, if 
\begin{equation}
\label{eq:cor1}
	\lim_{d\to\infty} \sum_{i=1}^d \frac{(\lambda_i^2-\tilde{\lambda}_i^2)^2}{\lambda_i^2 \tilde{\lambda}_i^2}(1-G_i^2) \quad \mbox{and} \quad \lim_{d\to\infty} \sum_{i=1}^d \hat{r}_i^2 \lambda_i^2 (1-G_i)
\end{equation}
are bounded and if 
\begin{equation}
\label{eq:cor2}
\lim_{d\to\infty} \sum_{i=1}^d \hat{r}_i^2 (\lambda_i^2-\tilde{\lambda}_i^2) (1-G_i^2)
\end{equation}
is bounded below, then $\mu > -\infty$.  
\end{lemma}

\begin{proof}
This result follows from the definition of $\mu$, which is equal to
$$
	\lim_{d \rightarrow \infty} \sum_{i=1}^d \left( -\frac{1}{2} \frac{(\lambda_i^2-\tilde{\lambda}_i^2)^2}{\lambda_i^2 \tilde{\lambda}_i^2}(1-G_i^2) - \hat{r}_i^2 \lambda_i^2 (1-G_i) + \frac{1}{2} \hat{r}_i^2 (\lambda_i^2-\tilde{\lambda}_i^2) (1-G_i^2)  \right),
$$
and our assumption that $\rho(G)<1$ so that $1-G_i^2>0$ and $1-G_i>0$.
\end{proof}

The terms in \eqref{eq:cor1} and \eqref{eq:cor2} provide the sense in which close agreement between the proposal equilibrium distribution $\normal(\mathcal{A}^{-1}\beta,\mathcal{A}^{-1})$ and target distribution $\normal(A^{-1}b,A^{-1})$, and gap between the spectrum of $G$ and $1$, imply a positive expected acceptance probability.

Note that \eqref{eq:cor2} is positive if $\lambda_i^2 > \tilde{\lambda}_i^2$, $\forall i$, and thus bounded below by zero.  This occurs when the proposal equilibrium distribution has greater variance (in every eigenvector direction) than the target Gaussian.  Thus, if $\hat{r}_i \neq 0$ then, adding the first term of \eqref{eq:cor1} to \eqref{eq:cor2}, we find that a proposal equilibrium distribution with greater variance than the target Gaussian increases the expected acceptance probability if the additional variance is not too great, i.e., if $(\lambda_i^2-\tilde{\lambda}_i^2)\lambda_i^{-2} \tilde{\lambda}_i^{-2} < \hat{r}_i$, $\forall i$.

Note also that the second term in \eqref{eq:cor1}, and \eqref{eq:cor2} are both zero if the means of the proposal equilibrium distribution and target agree, i.e., $\hat{r}_i=0$, $\forall i$.  

\subsection{Expected squared jump size for a Gaussian target}
\label{sec jumpsize}

The following theorem generalizes \cite[Prop. 3.8]{BPRSS2010}.

\begin{theorem}
\label{thm jumpsize}
Suppose that $G$ and $\Sigma$ are functions of $A$, and the Markov chain induced by~\eqref{MH0} is in equilibrium, i.e., $x \sim \normal(A^{-1}b,A^{-1})$.
With $\mu_{d,i}$ and $\sigma_{d,i}^2$ defined as in Theorem~\ref{thm accept}, let $q_i$ be a normalized eigenvector of $A$ corresponding to $\lambda_i^2$.  If there exists a $\delta > 0$ such that \eqref{eq Tcond} is satisfied, and $\mu^- := \lim_{d \rightarrow \infty} \sum_{j=1,j\neq i}^d \mu_{d,j}$ and $(\sigma^-)^2 := \lim_{d \rightarrow \infty} \sum_{j=1,j\neq i}^d \sigma_{d,j}^2$ exist, then \eqref{MH0} satisfies
\begin{equation}
\label{eq jump1}
	\mathrm{E}[(q_i^T (x'-x))^2]  = U_1 U_2 + E_3 + o(U_1) 
\end{equation}
as $d \rightarrow \infty$ where $|E_3| \leq U_3$,
\begin{align*}
	U_1 &= \tilde{g}_i^2 \hat{r}_i^2 + \frac{\tilde{g}_i^2}{\lambda_i^2} + \frac{\hat{g}_i}{\tilde{\lambda}_i^2}, \\
	U_2 &= \mathrm{E}[1 \wedge \mathrm{e}^{X}] = \Phi(\tsfrac{\mu^-}{\sigma^-}) + \mathrm{e}^{\mu^- + (\sigma^-)^2/2} \Phi(-\sigma^- - \tsfrac{\mu^-}{\sigma^-}), \; X \sim N(\mu^-,(\sigma^-)^2),\\
	U_3 &= ( \sigma_{d,i}^2 + \mu_{d,i}^2 )^{1/2}  
	\times \left( \tilde{g}_i^4 \hat{r}_i^4 + \frac{3}{\lambda_i^4} (\tilde{g}_i^2 + \tilde{r}_i \hat{g}_i)^2 + \frac{6}{\lambda_i^2} \hat{r}_i^2 \tilde{g}_i^2 ( \tilde{g}_i^2 + \tilde{r}_i \hat{g}_i) \right)^{1/2}.
\end{align*}
\end{theorem}

As discussed in the introduction, large expected squared jump size implies small first-order autocorrelation, which is a desirable property for a method to be efficient.  The following lemma is a special case of Theorem \ref{thm jumpsize} and provides a sufficient condition for a strictly positive expected jump size as $d \rightarrow \infty$, when the target is Gaussian.

\begin{lemma}
\label{lem jump}
With the same assumptions as Theorems \ref{thm accept} and \ref{thm jumpsize}, if $\mu_{d,i}$ and $\sigma_{d,i} \rightarrow 0$ as $d \rightarrow \infty$, then, as $d \rightarrow \infty$, \eqref{MH0} satisfies
\begin{equation}
\label{eq:lem1}
	\mathrm{E}[(q_i^T (x'-x))^2] - \left(\frac{1+G_i}{\tilde{\lambda}_i^2} + \frac{(1-G_i)}{\lambda_i^2} +  \hat{r}_i^2 \right) (1-G_i) \mathrm{E}[\alpha(x,y)] \rightarrow 0 
\end{equation}

In particular, if the conditions for Lemma \ref{lem accept} are met and $\lim_{d\rightarrow\infty} 1-G_i > 0$ then $\lim_{d\rightarrow\infty} \mathrm{E}[(q_i^T (x'-x))^2] >0$.  
\end{lemma}

Thus, in the special case when $\mu_{d,i}$ and $\sigma_{d,i} \rightarrow 0$ as $d \rightarrow \infty$, a strictly positive jump size requires a strictly positive expected acceptance probability \emph{and} $1-G_i$ bounded away from $0$.  For example, Lemma \ref{lem jump} can be applied to MALA with a Gaussian target which, after optimally tuning to maximize jump size, has a positive expected acceptance probability; however, the expected jump size goes to zero as $d \rightarrow \infty$ since $\lim_{d\rightarrow\infty} 1-G_i = 0$.  

When $\hat{r}_i \neq 0$ (the mean of the proposal equilibrium and the mean of the target are different) then the conflict between the action of the proposal (moving towards the mean of the proposal equilibrium) and the MH accept/reject step (favouring moves towards the mean of the target) increases the jump size.  This type of mixing may not increase overall efficiency.

Lemmas \ref{lem accept} and \ref{lem jump} show that $G_i \approx 1$ can imply small jumps in spite of a positive expected acceptance probability.  This is explained by noting that, in general, small jumps are accepted more frequently than large jumps.  A method with overly small jumps is not efficient.

Combining Lemmas \ref{lem accept} and \ref{lem jump} we see that an efficient method (with large jump size in all directions) will have close agreement between the eigenvalues of $\mathcal{A}$ and $A$ (the precision matrices of the proposal equilibrium and Gaussian target respectively) to maximize the expected acceptance probability, \emph{and} the spectrum of $G$ will be bounded away from $1$ to also maximize the expected squared jump size in every direction.  

Whether or not close agreement between the means of the proposal equilibrium and Gaussian target is desirable is less obvious from this theory and the precise effect of $\hat{r}_i$ depends on the values of $G_i$, $\lambda_i^2$ and $\tilde{\lambda}_i^2$.  Nevertheless, $\hat{r}_i= 0$, $\forall i$, does not preclude a method from having positive expected squared jump size.  

The terms in \eqref{eq:cor1}, \eqref{eq:cor2} and \eqref{eq:lem1} provide us with the sense in which an efficient method will have a proposal equilibrium distribution close to the target and the spectrum of $G$ bounded away from $1$.

A method for which the expected squared jump size is positive in the limit $d \rightarrow \infty$ is often referred to as dimension-independent because it requires only $\mathcal{O}(1)$ iterations to explore state space, once in equilibrium \cite{BS2009a,BPSSS2011,CRSW2013}.

\section{Non-Gaussian targets}
\label{sec nongaussian}

The results in Section \ref{sec gaussian} can be extended to non-Gaussian target distributions in 
some cases.  We follow the approach in \cite{BRS2009}.  

Consider the MH algorithm defined by
\begin{equation}
\label{MH0ng}
\begin{split}
	\mbox{Target:} & \qquad \pi_d \mbox{ defined by \eqref{eq target0} and \eqref{eq ref}}, \\
	\mbox{Proposal:} & \qquad y = G x + g + \nu, \qquad \mbox{where $\nu \iid \normal(0,\Sigma)$}. 
\end{split}
\end{equation}
The acceptance probability satisfies
$$
	\alpha(x,y) = 1 \wedge \exp\left( \phi_d(x)-\phi_d(y) + Z \right)
$$
where $Z = \log(\tsfrac{\tilde{\pi}_d(y) q(y,x)}{\tilde{\pi}_d(x) q(x,y)} )$.  Define $\tilde{\alpha}(x,y) = 1 \wedge \exp(Z)$ to be the acceptance probability for Algorithm~\eqref{MH0}.  We denote by $\mathrm{E}_{\pi_d}[\alpha(x,y)]$ the expectation of $\alpha(x,y)$ over $x \sim \pi_d$ and $y$ from \eqref{MH0ng}.  Similarly, $\mathrm{E}_{\tilde{\pi}_d}[\tilde{\alpha}(x,y)]$ is the expected acceptance probability of~\eqref{MH0} in equilibrium.

We associate with the precision matrix $A$, of the  Gaussian reference measure $\tilde{\pi}$, the norm $| \cdot |_s$ on $\bbR^d$, for any $s \in \bbR$, defined by
$$
	| x |_s = | A^s x |, \qquad \forall x \in \bbR^d.
$$
If $\lambda_1^2$ is the smallest eigenvalue of $A$, then
\begin{equation}
\label{eq:normbound}
	|x|_s \leq \lambda_1^{2(s-r)} |x|_r \qquad \forall s<r.
\end{equation}

\begin{assumption}
\label{assumption0}
Suppose there exists $M>0$ such that 
$$
	|\phi_d(x)| \leq M, \qquad \forall x \in \bbR^d.
$$
\end{assumption}

\begin{assumption}
\label{assumption1}
Suppose there exist constants $m,s,s',s'' \in \bbR$, $C,p>0$, and a locally bounded function $\delta: \bbR^+ \times \bbR^+ \mapsto \bbR^+$ such that for all sufficiently large $d$ and $\forall x,y \in \bbR^d$
\begin{align*}
	\phi_d(x) &\geq m, \\
	|\phi_d(x) - \phi_d(y)| &\leq \delta( |x-A^{-1}b|_s,|y-A^{-1}b|_s) \, |x-y|_{s'}, \\
	| \phi_d(x) | &\leq C (1 + |x-A^{-1}b|_{s''}^p ).
\end{align*}
\end{assumption}

\begin{assumption}
\label{assumption2}
Suppose that $r \in \mathbb{R}$ is such that 
$$
	\lim_{d \rightarrow \infty} \sum_{i=1}^d \lambda_i^{4r-2} < \infty.
$$
\end{assumption}

\subsection{Expected acceptance probability for a non-Gaussian target}
\label{sec:acc rate ng}

The following theorem applies to inverse problems with a Gaussian prior and bounded likelihood, and is similar to \cite[Thm. 2]{BRS2009} for RWM and SLA.

\begin{theorem}
\label{thm:ng12}
Suppose $\phi_d$ satisfies Assumption \ref{assumption0}, then Algorithm~\eqref{MH0ng} in equilibrium satisfies
\begin{align}
	c \mathrm{E}_{\tilde{\pi}_d}[\tilde{\alpha}(x,y)] &\leq \mathrm{E}_{\pi_d}[\alpha(x,y)] \leq C \mathrm{E}_{\tilde{\pi}_d}[\tilde{\alpha}(x,y)] \label{ng1}\\
	0 & < \mathrm{E}_{\pi_d}[\alpha(x,y)] \qquad \mbox{if $\mathrm{E}_{\tilde{\pi}_d}[|Z|] < \infty$} \label{ng2}
\end{align}
for constants $c=\mathrm{e}^{-3M}$ and $C = \mathrm{e}^{3M}$.
\end{theorem}

\begin{proof}  We follow the same reasoning as in the proof of \cite[Thm. 2]{BRS2009}.
Note that $\exp(-2M)(1 \wedge \exp(Z)) \leq 1 \wedge \exp( \phi_d(x)-\phi_d(y) + Z) \leq \exp(2M) (1 \wedge \exp(Z))$ and $\exp(-M) \tilde{\pi}_d(x) \leq \pi_d(x) \leq \exp(M) \tilde{\pi}_d(x)$.  Hence, we obtain~\eqref{ng1}.

To prove \eqref{ng2} first note for a random variable $X$ and any $\gamma > 0$ we have $\mathrm{E}[1 \wedge \exp(X)] \geq \exp(-\gamma)(1-\gamma^{-1} \mathrm{E}[|X|])$, see \cite[Lem. B1]{BRS2009}.  Also note that $C_0 := \mathrm{E}_{\pi_d}[|\phi_d(x)-\phi_d(y)+Z|] \leq C + C \mathrm{E}_{\tilde{\pi}_d}[|Z|] < \infty$.  Hence, we obtain \eqref{ng2} by taking $\gamma = 2 C_0$.
\end{proof}

Thus, in this weak sense, the expected acceptance probability of \eqref{MH0ng} with non-Gaussian target mimics the expected acceptance probability of \eqref{MH0} with a Gaussian target.  However, if $M$ is large Theorem \ref{thm:ng12} provides little useful information because the bounds in \eqref{ng1} are very loose.

The next theorem gives the acceptance probability for a non-Gaussian target more precisely. 

\begin{theorem}
\label{thm:ng1}
Suppose that $\phi_d$ and $r =\max \{ s,s',s'' \}$ satisfy Assumptions \ref{assumption1} and \ref{assumption2}, respectively.  Also suppose there exists $\{t_{d,i}\}$ and $t>0$ such that $t_{d,i} = \mathcal{O}(d^{-t})$ (uniformly in $i$) as $d \rightarrow \infty$ such that Algorithm~\eqref{MH0ng} satisfies
\begin{align*}
	&\mbox{$G$ and $\Sigma$ are functions of $A$}, \\
	&\tilde{g}_i^2\hat{r}_i^2 \lambda_i^2, \; \tilde{g}_i^2, \; \hat{g}_i \mbox{ are $\mathcal{O}(t_{d,i})$ (uniformly in $i$) as $d \rightarrow \infty$,} \\
	&\tilde{r}_i \mbox{ is bounded uniformly in $d$ and $i$}, \\
	&T_{1i}, T_{2i}, T_{3i}, T_{4i}, T_{5i} \mbox{ are } \mathcal{O}(d^{-1/2}) \mbox{ as $d \rightarrow \infty$ (uniformly in $i$), and} \\
	&\lim_{d\rightarrow\infty} \sum_{i=1}^d T_{1i}^2 + T_{2i}^2 + 2 T_{3i}^2 + 2 T_{4i}^2 + T_{5i}^2 < \infty.
\end{align*}
Let $q_i$ be a normalized eigenvector of $A$ corresponding to the eigenvalue $\lambda_i^2$, $\mu$ and $\sigma^2$ be as in Theorem~\ref{thm accept}, $\kappa_{d,i} = \mathrm{E}_{\pi_d}[q_i^T A^{1/2}(x - A^{-1}b)]$, and $\gamma_{d,i} = \mathrm{E}_{\pi_d}[(q_i^T A^{1/2}(x - A^{-1}b))^2]$.  If 
\begin{align*}
	\mu_{ng} &= \mu + \lim_{d \rightarrow \infty} \sum_{i=1}^d \kappa_{d,i} T_{1i} + T_{3i}(\gamma_i-1) \mbox{ and}  \\
	\sigma_{ng}^2 &= \sigma^2 + \lim_{d\rightarrow \infty} \sum_{i=1}^d  \left(\kappa_{d,i} T_{1i} + T_{3i} (\gamma_{d,i} - 1) \right)^2
\end{align*}
exist then
$$
	\mathrm{E}_{\pi_d}[\alpha(x,y)] \rightarrow \mathrm{E}[1 \wedge \mathrm{e}^{Z_{ng}}] = \Phi(\tsfrac{\mu_{ng}}{\sigma_{ng}}) + \mathrm{e}^{\mu_{ng} + \sigma_{ng}^2/2} \Phi(-\sigma_{ng} - \tsfrac{\mu_{ng}}{\sigma_{ng}})
$$
as $d \rightarrow \infty$, where $Z_{ng} \sim \normal(\mu_{ng},\sigma_{ng}^2)$.
\end{theorem}

\begin{corollary}
\label{cor:ng}
In addition to the conditions for Theorem \ref{thm:ng1}, if
$$
	\lim_{d \rightarrow \infty} \sum_{i=1}^d T_{1i}^2 + T_{3i}^2 = 0
$$
then 
$$
	\mu_{ng} = \mu \qquad \mbox{and} \qquad \sigma_{ng}^2 = \sigma^2,
$$
and the expected acceptance probability for the non-Gaussian target case has the same limit as $d \rightarrow \infty$ as the Gaussian target case.
\end{corollary}

Unfortunately, to extend Theorem \ref{thm accept} to non-Gaussian targets we have had to make quite strong assumptions about the proposal and the size of the perturbation from a Gaussian target distribution.  These assumptions effectively limit the proposal to having zero jump size in the limit $d \rightarrow \infty$.  Nevertheless, in certain circumstances the expected acceptance probability of a MH algorithm with a stochastic AR(1) proposal is the same for both Gaussian and non-Gaussian targets.

\subsection{Expected squared jumpsize for non-Gaussian target}

Using a similar proof structure to Theorem \ref{thm:ng12} we obtain the following result.
\begin{theorem}
\label{thm:ng3}
Suppose $\phi_d$ satisfies Assumption \ref{assumption0} and let $q \in \mathbb{R}^d$.  Let $\mathrm{E}_{\pi_d}[(q^T(x'-x))^2]$ and $\mathrm{E}_{\tilde{\pi}_d}[(q^T(x'-x))^2]$ denote the expected squared jump size in direction $q$ of MH algorithms \eqref{MH0ng} and \eqref{MH0} in equilibrium, respectively.  Then 
$$
	c \mathrm{E}_{\tilde{\pi}_d}[(q^T(x'-x))^2] \leq \mathrm{E}_{\pi_d}[(q^T(x'-x))^2] \leq C \mathrm{E}_{\tilde{\pi}_d}[(q^T(x'-x))^2] 
$$
for constants $c=\mathrm{e}^{-3M}$ and $C = \mathrm{e}^{3M}$.
\end{theorem}

Thus, in this weak sense, the jump size of \eqref{MH0ng} for non-Gaussian targets mimics the jump size of \eqref{MH0} for Gaussian targets; however, if $M$ is large then the jump size of \eqref{MH0ng} and \eqref{MH0} could be quite different.

The following corollary is a consequence of Theorem \ref{thm:ng3}, Lemma \ref{lem accept} and Corollary \ref{cor:accept} and provides sufficient conditions for \eqref{MH0ng} to be dimension independent.  

\begin{corollary}
\label{cor:dimindep}
Suppose $\phi_d$ satisfies Assumption \ref{assumption0} uniformly in $d$, the proposal of \eqref{MH0ng} satisfies the assumptions of Theorem \ref{thm accept}, \eqref{eq:cor1} are bounded, and \eqref{eq:cor2} is bounded below.

Then there exists $c>0$ such that for all sufficiently large $d$ and $\forall q \in \mathbb{R}^d$, Algorithm \eqref{MH0ng} in equilibrium satisfies
$$
	\mathrm{E}_{\pi_d}[(q^T(x'-x))^2] \geq c >0.
$$
\end{corollary}

The following theorem finds the expected squared jump size of MH algorithms \eqref{MH0ng} more generally.

\begin{theorem}
\label{thm:ng2}
Under the same conditions as Theorem \ref{thm:ng1}, 
\begin{multline*}
	\mathrm{E}_{\pi_d}[ (q_i^T (x' - x))^2 ] \\= \left( \left(\tilde{g}_i^2 \hat{r}_i^2 + \frac{\hat{g}_i}{\tilde{\lambda}_i^2} \right) \mathrm{E}_{\pi_d}[\alpha(x,y)] + 2 \frac{\hat{r}_i \tilde{g}_i^2 \gamma_{d,i}^{1/2}}{\lambda_i} u_{d,i} + \frac{\tilde{g}_i^2 \gamma_{d,i}}{\lambda_i^2} v_{d,i} \right) + \mathrm{o}(t_{d,i} \lambda_i^{-2}) 
\end{multline*}
(uniformly in $i$) as $d \rightarrow \infty$, for some $-1 \leq u_{d,i} \leq 1$ and $0 \leq v_{d,i} \leq 1$.
\end{theorem}

Despite all of the terms of the expected squared jump size going to zero as $d \rightarrow \infty$ in the above theorem (by the conditions of the theorem), the leading order term is $\mathcal{O}(t_{d,i} \lambda_i^{-2})$. 

\section{Examples}
\label{sec examples}

\subsection{Discretized Langevin diffusion}
\label{sec langevin}

For s.p.d. matrix $V \in \bbR^{d \times d}$ (the `preconditioner'), $\theta \in [0,1]$, and time step $h >0$, we can discretize Langevin diffusion to construct a proposal by 
$$
	y = x + \tsfrac{h}{2} V \nabla \log \pi_d (\theta y + (1-\theta)x) + \sqrt{h} \nu
$$
where $\nu \iid \normal(0,V)$.  This proposal can be simplified by replacing $\pi_d$ with $\tilde{\pi}_d$ to obtain a MH algorithm with a stochastic AR(1) proposal
\begin{equation}
\label{GenLang}
\begin{split}
	\mbox{Target:} & \; \pi_d, \\
	\mbox{Proposal:} & \; y = (I + \tsfrac{\theta h}{2} VA)^{-1} \left[ (I - \tsfrac{(1-\theta)h}{2} VA)x + \tsfrac{h}{2} V b + (h V)^{1/2} \xi \right],
\end{split}
\end{equation}
for $\xi \iid \normal(0,I)$.
No simplification is necessary when $\phi_d = 0$ and the target is Gaussian.  The proposal is a convergent stationary stocastic AR(1) process. Applying Theorem \ref{thm:2.1} gives $\mathcal{A} = A + (\theta-\tsfrac{1}{2}) \tsfrac{h}{2} AVA$ and $\mathcal{A}^{-1} \beta = A^{-1}b$.

Many existing MCMC methods correspond to different choices of $\theta$ and $V$, see Table~\ref{tab1}.

\begin{table}
\begin{tabular}{|c|c|c|p{7.2cm}|}
\hline
Target & $\theta$ & $V$ & Method \\
\hline \hline
Gaussian & $0$ & $I$ & MALA and SLA \cite{BRS2009} (the proposal chain is ULA \cite{RT1996}). \\ \hline
Non-Gaussian & $0$ & $I$ & SLA \cite{BRS2009}. \\ \hline
Gaussian & $0$ & $A^{-1}$ & Version of MALA used in \cite[eqn. 2.17]{PST2012}.  Preconditioned Simplified Langevin Algorithm (P-SLA) \cite{BRS2009}. \\ \hline
Non-Gaussian & $0$ & $A^{-1}$ & Preconditioned Simplified Langevin Algorithm (P-SLA) \cite{BRS2009}. \\ \hline
Non-Gaussian & $\in [0,1]$ & $I$ & $\theta$-SLA \cite{BRS2009}.  See also \cite[eqn. 6.2]{CRSW2013}.  \\ \hline
Non-Gaussian & $0.5$ & $I$ & CN \cite{CRSW2013}. \\ \hline
Gaussian & $0.5$ & $A^{-1}$ & pCN \cite{CRSW2013}.   Scaled Stochastic Newton \cite{BG2014}. \\ \hline
Non-Gaussian & $0.5$ & $A^{-1}$ & pCN \cite{CRSW2013}. \\ \hline
\end{tabular}
\caption{Different choices of $\theta$ and $V$ in \eqref{GenLang} lead to different MCMC methods.  When the target is Gaussian, then a local Gaussian approximation to the target is global and non-stationary proposals become stationary, hence some methods coincide when the target is Gaussian.  Other choices of $\theta$ and $V$ are possible.}
\label{tab1}
\end{table}

Matrices $G$ and $\Sigma$ in \eqref{GenLang} are not functions of the $A$, Theorems \ref{thm:ng1} or \ref{thm:ng2} do not apply, but a coordinate transformation will fix this!  The proof of the following lemma is straightforward, so is omitted.

\begin{lemma}
\label{lem ch1}
The coordinate transformation
$$
	\hat{x} = V^{-1/2} x
$$
transforms Algorithm~\eqref{GenLang} to the MH algorithm defined by
\begin{equation}
\label{eq hat1}
\begin{split}
	\mbox{Target:} & \; \pi_d \mbox{ where } \tsfrac{\dd \pi_d}{\dd \tilde{\pi}_d}(x) = \exp( - \psi(x) ) \mbox{ and }\tilde{\pi}_d \mbox{ is }  \normal(B^{-1} V^{1/2} b, B^{-1}), \\
	\mbox{Proposal:} & \;  y = (I + \tsfrac{\theta h}{2} B)^{-1} \left[ (I - \tsfrac{(1-\theta)h}{2} B) x + \tsfrac{h}{2} V^{1/2} b + h^{1/2} \xi \right], 
\end{split}
\end{equation}
where $\psi(x) = \phi_d(V^{1/2}x)$, $B = V^{1/2} A V^{1/2}$ and $\xi \iid \normal(0,I)$.
\end{lemma}

Matrices $G$ and $\Sigma$ in \eqref{eq hat1} are now functions of $B$ so we can apply Theorems~\ref{thm:ng1} and \ref{thm:ng2} to \eqref{eq hat1} to obtain the following theorem.  Algorithm~\eqref{eq hat1} is not intended for computation; we only use it to determine the convergence properties of \eqref{GenLang}.  The proof of Theorem~\ref{thm genlang conv} is in the appendix; further details can be found in \cite{NFelecreport}.

\begin{theorem}
\label{thm genlang conv}
Suppose there are constants $c,C > 0$ and $\kappa \geq 0$ such that the eigenvalues $\lambda_i^2$ of $B = V^{1/2} A V^{1/2}$ (equivalently, $\lambda_i^2$ are eigenvalues of $VA$) satisfy
\[
	c i^\kappa \leq \lambda_i \leq C i^\kappa, \qquad i=1,\ldots,d.
\]
Also suppose that $\psi_d(x) := \phi_d(V^{1/2}x)$ satisfies Assumption \ref{assumption1} (with $|\cdot|_s = |B^s \cdot|$ for $s \in \bbR$) and $r = \max \{ s,s',s'' \}$ satisfies Assumption \ref{assumption2}.

If $h = l^2 d^{-1/3 - 2 \kappa}$ for $l > 0$ and $\tau = \lim_{d\rightarrow \infty} \frac{1}{d^{6\kappa + 1}} \sum_{i=1}^d \lambda_i^6$ then Algorithm~\eqref{GenLang} in equilibrium satisfies
\begin{equation}
\label{eq expect}
	\mathrm{E}[\alpha(x,y)] \rightarrow 2 \Phi\left( -\frac{l^3 |\theta-\tsfrac{1}{2}| \sqrt{\tau}}{4} \right) 
\end{equation}
and for normalized eigenvector $q_i$ of $B$ corresponding to $\lambda_i^2$,
\begin{equation}
\label{eq jump2}
	\mathrm{E}[ |q_i^T V^{-1/2} (x' - x)|^2 ] = 2 h \Phi\left( -\frac{l^3 |\theta-\tsfrac{1}{2}| \sqrt{\tau}}{4} \right) + \mathrm{o}(h)
\end{equation}
as $d \rightarrow \infty$.
\end{theorem}

When $\theta \neq 1/2$,~\eqref{eq jump2} can be maximized by tuning $h$ (by tuning $l$).  Using $s^3 = l^3 |\theta-\tsfrac{1}{2}| \sqrt{\tau} / 4$, and ignoring the $\mathrm{o}(h)$ term, we have 
\begin{equation}
\label{eq:mj}
	\max_{l > 0} \frac{2 l^2}{d^{1/3+2\kappa}} \Phi \! \left( -\frac{l^3 |\theta-\tsfrac{1}{2}| \sqrt{\tau}}{4} \right) \! = \max_{s > 0} \left( \frac{128}{|\theta-\tsfrac{1}{2}|^{2} \tau  d^{1+6\kappa}} \right)^{\frac{1}{3}} \!\! s^2 \Phi(-s^3),
\end{equation}
which is maximized at $s = s_0=0.8252$, independent of $\tau$ and $\theta$.  Therefore, the acceptance probability that maximizes the expected jump distance is $2 \Phi(-s_0^3) = 0.574$.  This result was first stated for MALA in \cite{RR1998} for product-form target distributions, and generalized in \cite{RR2001,BRS2009}.  This result is a further generalization because it allows  off-diagonal covariance terms in the Gaussian reference measure.  The Gaussian case with off-diagonal covariance is also considered in \cite{BS2009a}, where it is assumed that the spectral decomposition of $A$ is available to explicitly perform computations in the coordinate system where the Gaussian has product form.  We stress that Algorithm~\ref{GenLang} is implemented in the original coordinate system and the spectral decomposition of $A$ never needs to be computed.

Even though the optimal expected acceptance probability is independent of $\tau$ and $\theta$, the jump size does depend on these quantities.  Jump size can be improved by choosing $V$ to control the eigenvalues of $B$ and minimize $\tau$, which is a scaled $6$-norm of the sequence $\{\lambda_i\}$ of eigenvalues of $B$.  The dependence of $\tau$ on $\{\lambda_i\}$ shows the relative importance of controlling small and large $\lambda_i$.   Other results in \cite{BS2009a,CRSW2013,PST2012} consider only $V = I$ or $A^{-1}$, i.e.,  no preconditioner or the perfect preconditioner.  Choosing $V = A^{-1}$ is typically computationally infeasible for large $d$, so it is useful to consider general s.p.d. $V$.  

We can also improve the jump size by choosing $\theta \approx \tsfrac{1}{2}$; however, if the action of $(I + \tsfrac{\theta h}{2} VA)^{-1}$ is expensive to compute then the computational cost of choosing $\theta \neq 0$ may outweigh the benefit of improved jump size.  

If $\theta = \tsfrac{1}{2}$, Theorem~\ref{thm genlang conv} still applies but the choice of $h = \mathcal{O}(d^{-1/3-2\kappa})$ is suboptimal because it implies a zero expected jump size in the limit $d \rightarrow \infty$.  Alternatively, $\theta = \tsfrac{1}{2}$ implies $Z = \log ( \frac{\tilde{\pi}_d(y)q(y,x)}{\tilde{\pi}_d(x) q(x,y)} ) = 0$, and, provided $\phi_d$ is bounded, we can apply Theorem \ref{thm:ng12} to obtain a strictly positive expected acceptance probability for any $h$.   Cotter \emph{et al.}~\cite{CRSW2013} have shown it is possible to obtain a strictly positive expected acceptance probability and jump size with positive $h$ as $d \rightarrow \infty$.  In this sense, Theorem~\ref{thm genlang conv} is slightly weaker than the results in \cite{CRSW2013} for the case $\theta = \tsfrac{1}{2}$, since our result does not prevent a zero expected acceptance probability as $d \rightarrow \infty$.  If $\theta = \tsfrac{1}{2}$ and the target is Gaussian then acceleration is possible \cite{F2013,FP2016,FP2014}.

\subsection{Multi-step proposals}
\label{sec lstep}

Given a stochastic AR(1) proposal~\eqref{eq:ar1} for some $G$, $g$ and $\Sigma$, we can form a new stochastic AR(1) proposal by taking $L$ steps
$$
	y^{(l)} = G y^{(l-1)} + g + \nu^{(l)}, \qquad l=1,\dotsc,L,
$$
where $\nu^{(l)} \stackrel{\text{i.i.d.}}{\sim} \normal(0,\Sigma)$ and $y^{(0)} = x$.  This yields a new \emph{multi-step} stochastic AR(1) proposal 
\begin{equation}
\label{eq:lprop}
	y = G_L x + g_L + \nu, \qquad \mbox{with } \nu \iid \normal(0,\Sigma_L)
\end{equation}
where $G_L = G^L$, $g_L = (I-G)^{-1}(I-G^L)g$ and $\Sigma_L = \sum_{l=0}^{L-1} G^l \Sigma (G^T)^l$.  Hence, the eigenvalues of $G_L$ are $G_i^L$, and if $G_i < 1$ then the multi-step proposal chain is convergent in distribution with the same equilibrium distribution as the $1$-step proposal chain, i.e., $\mathcal{A}_L = \mathcal{A}$ and $\beta_L = \beta$.  

When the proposal chain is $\normal(\mathcal{A}^{-1}\beta,\mathcal{A}^{-1})$-reversible, see Lemma~\ref{lem reversible}, the multi-step proposal MH algorithm corresponds to a surrogate transition method \cite[p.194]{Liu2001book} and the MH acceptance probability has simplified form and low computational cost, as shown in the following lemma.
\begin{lemma}
\label{lem:lstepalpha}
If $G$ and $\Sigma$ are functions of $A$, then the acceptance probability for a multi-step proposal satisfies
$$
	\alpha(x,y) 
	= 1 \wedge \frac{ \pi_d(y) q_L(y,x)}{ \pi_d(x) q_L(x,y) }
	= 1 \wedge \frac{ \pi_d(y) \pi^*(x)}{\pi_d(x) \pi^*(y)}
$$
where $q_L(x,dy) = q_L(x,y) \dd y$ is the transition kernel for a multi-step proposal $y$ given $x$ from \eqref{eq:lprop} and $\pi^*(x) \propto \exp( -\frac{1}{2} x^T \mathcal{A} x + \beta^T x )$.  
\end{lemma}

The computational cost of a multi-step proposal is $L$ times the cost of the single-step proposal, but the expected squared jump size for a MH algorithm with a multi-step proposal is, in general, not $L$ times the original.  

For example, consider multi-step SLA with $G = (I-\tsfrac{h}{2}A)$, $\beta = \tsfrac{h}{2}b$, and $\Sigma = hI$, and the proposal is given by \eqref{eq:lprop}.  Applying Theorems \ref{thm:ng1} and \ref{thm:ng2} and Corollary \ref{cor:ng} we obtain the following result.  
\begin{theorem}
\label{thm:lstep}
Suppose there exist constants $c,C>0$ and $\kappa \geq 0$ such that the eigenvalues $\lambda_i^2$ of $A$ satisfy
$$
	c i^\kappa \leq \lambda_i \leq C i^\kappa, \qquad i=1,\dotsc,d.
$$
Also suppose that $\phi_d$ satisfies Assumption \ref{assumption1} and $r = \max \{ s,s',s'' \}$ satisfies Assumption \ref{assumption2}.

If $h = l^2 d^{-1/3-2\kappa}$ for some $l >0$ then multi-step SLA, in equilibrium, satisfies
\begin{equation}
\label{eq:lstepa}
	\mathrm{E}[\alpha(x,y)] \rightarrow 2 \Phi\left( - \frac{l^3 \sqrt{L\tau}}{8} \right)
\end{equation}
and
\begin{equation}
\label{eq ljump}
	\mathrm{E}[(x_i'-x_i)^2] = 2 L h \Phi\left( -\frac{l^3 \sqrt{L\tau}}{8} \right) + \mathrm{o}(h)
\end{equation}
as $d \rightarrow \infty$ where $\tau = \lim_{d \rightarrow \infty} \frac{1}{d^{1 + 6 \kappa}} \sum_{i=1}^d \lambda_i^6$.
\end{theorem}

The performance of multi-step SLA can be maximized by tuning $l$ to maximize the expected jump size.  From \eqref{eq ljump}, using $s = l (L\tau)^{1/6}/2$, we have
\begin{equation}
\label{eq:lmj}
	\max_{l > 0} \frac{ 2 L l^2}{ d^{1/3 + 2 \kappa}} \Phi\left( -\frac{l^3 \sqrt{L\tau}}{8} \right) = \max_{s>0} \left( \frac{512 L^{2}}{ d^{1 + 6 \kappa} \tau } \right)^{\frac{1}{3}} s^2 \Phi( - s^3 ),
\end{equation}
which is maximized at $s=s_0 = 0.8252$.  Therefore, the expected jump size of multi-step SLA is maximized when the acceptance probability is $2 \Phi(-s_0^3) = 0.574$, which is the same as SLA, but this corresponds to an expected jump size that is $L^{2/3}$ times larger than the jump size for SLA (compare \eqref{eq:lmj} and \eqref{eq:mj} with $\theta=0$) in the limit $d \rightarrow \infty$.

To compare the efficiency of multi-step SLA for varying $L$ we must also consider the computational cost of the method.  For example, suppose that matrix-vector products with $A$ cost $1$ unit of CPU time, inner products and drawing independent samples from $\normal(0,I)$ have negligible cost, and evaluating $\phi_d$ costs $t$ units of CPU time.  From Lemma \ref{lem:lstepalpha}, the acceptance ratio for multi-step SLA is
$$
	\alpha(x,y) = 1 \wedge \exp\left( \tsfrac{h}{4} (|A x|^2 - |Ay|^2) - \tsfrac{h}{2} b^T (Ax-Ay) + \phi_d(x) - \phi_d(y) \right),
$$
so multi-step SLA uses $L$ matrix vector products with $A$ per proposal and an additional matrix-vector product and two evaluations of $\phi_d$ in the acceptance ratio.  If the proposal is accepted then we can reuse some of the calculations in the acceptance ratio, but if it is rejected then a matrix-vector product and an evaluation of $\phi_d$ are wasted.  The average cost of a multi-step SLA iteration is then
$$
	L + t + (1-\alpha)(1 + t) = 1.426 + 0.426 t + L
$$
units of CPU time, after optimally tuning multi-step SLA so that the acceptance probability is $0.574$.  Also let $1$ unit of jump size be the expected jump size of SLA, then multi-step SLA has an expected jump size of $L^{2/3}$ units, and the `efficiency' of multi-step SLA is calculated as jump size divided by CPU time,
$$
	\frac{L^{2/3}}{1.426 + 0.426 t + L},
$$
which is maximized at $L = 2(1.426 + 0.426 t)$.  Therefore, SLA can be improved by using multi-step SLA with $L > 1$, with the optimal value of $L$ increasing with the cost of evaluating $\phi_d$.  If $t = 0$, then $L=3$ is optimal.  Figure \ref{fig:1} also shows the efficiency of multi-step SLA for other values of $t$, showing that using $L\geq 3$ always improves efficiency.  This type of analysis can be repeated for other multi-step proposals.

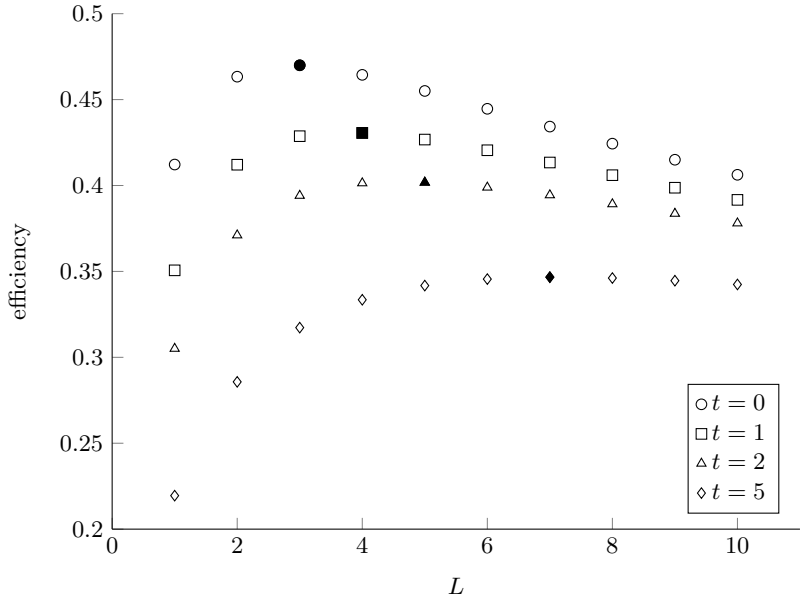
\begin{figure}
\begin{center}
%
%
%
%
\begin{tikzpicture}

\begin{axis}[%
view={0}{90},
width = 3.6in,
height = 2.7in,
scale only axis,
xmin=0, xmax=11,
xlabel={$L$},
ymin=0.2, ymax=0.5,
ylabel={efficiency},
axis lines*=left,
legend style={at={(0.97,0.03)},anchor=south east,align=left}]
\addplot [
color=black,
only marks,
mark=o,
mark options={solid}
]
coordinates{
 (1,0.412201154163232)(2,0.463339478099299)(3,0.46996923250156)(4,0.464401419054505)(5,0.455029215408165)(6,0.444644121854919)(7,0.434287409212316)(8,0.424358158285593)(9,0.414996039796876)(10,0.406230424786695) 
};
\addlegendentry{$t=0$};

\addplot [
color=black,
only marks,
mark=square,
mark options={solid}
]
coordinates{
 (1,0.350631136044881)(2,0.412097884726947)(3,0.428706476309131)(4,0.430595027305151)(5,0.426739307970354)(6,0.420520536028353)(7,0.413387450296314)(8,0.406008932196508)(9,0.398705188990253)(10,0.391629162471547) 
};
\addlegendentry{$t=1$};

\addplot [
color=black,
only marks,
mark=triangle,
mark options={solid}
]
coordinates{
 (1,0.305064063453325)(2,0.371061489473632)(3,0.394104551544506)(4,0.401376568937519)(5,0.401761162161702)(6,0.398879831951513)(7,0.39440673744589)(8,0.389180774469741)(9,0.383645035549053)(10,0.378041116925621) 
};
\addlegendentry{$t=2$};

\addplot [
color=black,
only marks,
mark=diamond,
mark options={solid}
]
coordinates{
 (1,0.219490781387182)(2,0.285709332607667)(3,0.317279411691871)(4,0.333488896213571)(5,0.34175055378832)(6,0.345534454677127)(7,0.346656471203389)(8,0.346140533056421)(9,0.344596106317476)(10,0.342401064739804) 
};
\addlegendentry{$t=5$};

\addplot [
color=blue,
only marks,
mark=*,
mark options={solid,fill=black,draw=black},
forget plot
]
coordinates{
 (3,0.46996923250156) 
};
\addplot [
color=blue,
only marks,
mark=square*,
mark options={solid,fill=black,draw=black},
forget plot
]
coordinates{
 (4,0.430595027305151) 
};
\addplot [
color=blue,
only marks,
mark=triangle*,
mark options={solid,fill=black,draw=black},
forget plot
]
coordinates{
 (5,0.401761162161702) 
};
\addplot [
color=blue,
only marks,
mark=diamond*,
mark options={solid,fill=black,draw=black},
forget plot
]
coordinates{
 (7,0.346656471203389) 
};
\end{axis}
\end{tikzpicture}%
\end{center}
\caption{Efficiency of multi-step SLA for varying number of steps $L$ and varying computing cost $t$ for evaluating $\phi_d$.  Filled markers correspond to optimal efficiency for given $t$. }
\label{fig:1}
\end{figure}

\subsection{Hybrid Monte Carlo}

HMC \cite{DKPR1987,BPRSS2010,N1993} with a normal target is another example of a stochastic AR(1) proposal.  When the target is a perturbation of a Gaussian ($\pi_d$ given by \eqref{eq target0} and \eqref{eq ref})  we can replace $\nabla \log \pi_d(x)$ in the definition of HMC with $\nabla \log \tilde{\pi}_d(x) = -Ax+b$ and also obtain a stochastic AR(1) proposal.  This is analogous to simplifying MALA to obtain SLA.

Choosing a s.p.d. matrix $V$, a step size $h$, and a number of steps $L$, the HMC proposal we analyze is computed by defining $q_0 = x$, sampling $p_0 \iid \normal(0,V^{-1})$, then for $l=0,\dotsc,L-1$ computing
\begin{align*}
	p_{l+1/2} &= p_{l} - \tsfrac{h}{2} (A q_{l} - b), \\
	q_{l+1} &= q_{l} + h V p_{l+1/2}, \\
	p_{l+1} &= p_{l+1/2} - \tsfrac{h}{2} (A q_{l+1} - b).
\end{align*}
The proposal is then $y := q_L$  given by 
\begin{equation}
\label{hmc prop2}
	\left[ \twobyone{y}{p_L} \right]
	=
	K^L \left[ \twobyone{x}{\xi} \right] + \sum_{l=0}^{L-1} K^l J \left[ \twobyone{0}{\tsfrac{h}{2}b} \right] 
	\quad\mbox{with $\xi \iid \normal(0,V)$},
\end{equation}
where $K,J \in \bbR^{2d \times 2d}$ are defined as
$$
	K = 
	\left[ \!\! \twobytwo{I-\tsfrac{h^2}{2}VA}{hV}{-h A + \tsfrac{h^3}{4}AVA}{I-\tsfrac{h^2}{2}AV} \!\right]
\mbox{ and }
	J = 	
	\left[ \twobytwo{2I}{hV}{-\tsfrac{h}{2}A}{2I - \tsfrac{h^2}{2}AV} \right],
$$
which can also be written as
\begin{equation}
\label{hmc prop}
	y = (K^L)_{11} x + \left( SJ \left[ \twobyone{0}{\tsfrac{h}{2}b} \right] \right)_1 + (K^L)_{12} \xi, \quad \mbox{with $\xi \iid \normal(0,V^{-1})$,}
\end{equation}
where $(K^L)_{ij}$ is the $ij$ block (of size $d\times d$) of $K^L$, $S = (I-K)^{-1} (I-K^L)$ and $(\cdot)_1$ are the first $d$ entries of the vector $(\cdot)$.

Therefore, applying Theorem \ref{thm:2.1} and Corollary \ref{lem equiv}, $\mathcal{A} = \Sigma^{-1} (I-G^2)$ and $\mathcal{A}^{-1} \beta = A^{-1} b$, where $G = (K^L)_{11}$ and $\Sigma = (K^L)_{12} V^{-1} (K^L)_{12}^T$.

For this version of HMC, $G$ and $\Sigma$ are not functions of $\mathcal{A}$, as required for our theory.  Again, a coordinate transformation will fix this.

\begin{theorem}
\label{thm iso hmc}
Under the coordinate transformation
$$
	\left[ \twobyone{\hat{x}}{\hat{p}} \right] \leftrightarrow \mathcal{V}^{-1} \left[ \twobyone{x}{p} \right], \qquad \mbox{where } 
	\mathcal{V} = \left[ \twobytwo{V^{1/2}}{0}{0}{V^{-1/2}} \right] \in \bbR^{2d \times 2d},
$$
the Hamiltonian $H(x,p):= \frac{1}{2} p^T V p + \frac{1}{2} x^T A x - b^T x$ and MH algorithm with target $\pi_d$ and proposal \eqref{hmc prop} are transformed to Hamiltonian $\mathcal{H}(x,p):= \tsfrac{1}{2} p^T p + \tsfrac{1}{2} x^T B x - (V^{1/2}b)^T x$ and a MH algorithm with
\begin{equation}
\label{eq hmchat1}
\begin{split}
	\mbox{Target:} & \; \pi_d \mbox{ where } \tsfrac{\dd \pi_d}{\dd \tilde{\pi}_d}(x) = \exp( - \psi(x) ), \tilde{\pi}_d \mbox{ is }  \normal(B^{-1} V^{1/2} b, B^{-1}), \\
	\mbox{Proposal:} & \; 
	y = (\mathcal{K}^L)_{11} x + \left( \mathcal{S} \mathcal{J} \left[ \twobyone{0}{\tsfrac{h}{2} V^{1/2} b} \right] \right)_1 + (\mathcal{K}^L)_{12} \xi,
\end{split}
\end{equation}
where $\xi \sim \normal(0,I)$, $\psi(x) = \phi_d(V^{1/2}x)$, $B = V^{1/2} A V^{1/2}$, $\mathcal{S} = (I-\mathcal{K})^{-1}(I-\mathcal{K}^L)$,
$$
	\mathcal{K} = \left[ \twobytwo{I-\tsfrac{h^2}{2}B}{hI}{-h B + \tsfrac{h^3}{4}B^2}{I-\tsfrac{h^2}{2}B} \right]
	\quad \mbox{and} \quad
	\mathcal{J} = \left[ \twobytwo{2I}{hI}{-\tsfrac{h}{2}B}{2I - \tsfrac{h^2}{2}B} \right].
$$
\end{theorem}

\begin{proof}
Use $K = \mathcal{V} \mathcal{K} \mathcal{V}^{-1}$ and $J = \mathcal{V} \mathcal{J} \mathcal{V}^{-1}$.
\end{proof}

Similar coordinate transformations are used in classical mechanics \cite[p. 103]{A1989}, see also \cite{BCSS2014}.

Note that the transformed MH algorithm has $G$ and $\Sigma$ as functions of $B$ and $\mathcal{A} = (\mathcal{K}^L)_{12}^{-2} (I - (\mathcal{K}^L)_{11}^2)$ and $\mathcal{A}^{-1} \beta = B^{-1} V^{1/2} b$.  Therefore, we can apply Theorems \ref{thm:ng1} and \ref{thm:ng2} to this version of HMC.  First, we need to determine the eigenvalues of $G = (\mathcal{K}^L)_{11}$, which are the same as the eigenvalues of $(K^L)_{11}$ since the matrices are similar.  The proof of the following result is in the appendix.
\begin{theorem}
\label{thm uhmc conv}
Let $\lambda_i^2$ be eigenvalues of $B = V^{1/2} A V^{1/2}$.  Then $G = (\mathcal{K}^L)_{11}$ has eigenvalues 
$$
	G_i = \cos( L \theta_i )
$$
where $\theta_i = -\cos^{-1} ( 1 - \tsfrac{h^2}{2} \lambda_i^2 )$.
\end{theorem}

Using these two theorems, we could apply Theorems \ref{thm:ng1} and \ref{thm:ng2} to HMC, when $L$ is fixed as $d \rightarrow \infty$.  Instead, as in \cite{BPRSS2010}, we consider the case when $Lh$ is kept fixed while $d \rightarrow \infty$.  For this latter case, HMC does not satisfy all of the conditions for Theorems \ref{thm:ng1} and \ref{thm:ng2}, since $\mathrm{E}_{\pi_d}[|y - x|_r^{2q}] \nrightarrow 0$ as $d \rightarrow \infty$, since $\tilde{g}_i^2, \; \hat{g}_i \nrightarrow 0$ as $d \rightarrow \infty$, see Lemma \ref{lem:phi} required for the proof of Theorems \ref{thm:ng1} and \ref{thm:ng2}.  However, if we consider the normal target case ($\phi(x)=0$), then we can apply Theorem \ref{thm accept} and Lemma \ref{lem jump} to get the following result.

\begin{theorem}
\label{thm hmc conv}
Suppose there are constants $c,C>0$ and $\kappa \geq 0$ such that the eigenvalues of $B = V^{1/2} A V^{1/2}$ satisfy
$$
	c i^\kappa \leq \lambda_i \leq C i^\kappa, \qquad i=1,\dotsc,d.
$$
If $h = l d^{-1/4-\kappa}$ for $l > 0$, and $L = \lfloor \tsfrac{T}{h} \rfloor$ for fixed $T$, then HMC (with proposal \eqref{hmc prop} and target $\normal(A^{-1}b,A^{-1})$), in equilibrium, satisfies
\begin{equation}
\label{eq hmcexpect}
	\mathrm{E}[\alpha(x,y)] \rightarrow  2 \Phi\left( -\frac{l^2 \sqrt{\tau}}{8} \right) 
\end{equation}
where $\tau = \lim_{d \rightarrow \infty} \frac{1}{d^{1+4\kappa}} \sum_{i=1}^d \lambda_i^4 \sin^2(\lambda_i T')$
and for eigenvector $q_i$ of $B$ corresponding to $\lambda_i^2$,
\begin{equation}
\label{eq hmcjump}
	\mathrm{E}[ |q_i^T V^{-1/2} (x' - x)|^2 ] \rightarrow 2 \frac{1-\cos(\lambda_i T')}{\lambda_i^2} 2 \Phi\left( -\frac{l^2 \sqrt{\tau}}{8} \right)
\end{equation}
as $d \rightarrow \infty$, where $T' = Lh$.
\end{theorem}

This result extends results in \cite{BPSSS2011,BPRSS2010} that only cover target distributions with diagonal covariance, and $\kappa > 1/2$ or $\kappa = 0$.  

Using the above result we can extend the cases for which $0.651$ is the optimal acceptance probability, obtained by by tuning $h$ (equivalently, by tuning $l$).  To maximize efficiency of HMC we should choose $h$ to maximize the jump size divided by compute time per proposal.  Since the compute time per proposal is proportional to $L = \frac{T'}{h}$, this corresponds to choosing $l$ to maximize
$$
	 C \sqrt{s} \Phi(s) 
$$
where $s = \frac{l^2 \sqrt{\tau}}{8}$ and $C$ is a constant, which is maximized at $s = s_0 = 0.4250$, which corresponds to an expected acceptance probability of $2 \Phi(s_0) = 0.651$.  This is the same expected acceptance probability found in \cite{BPRSS2010} that considered product form target distributions with $\lambda_i$ constant, $\forall i$, i.e., $d$ i.i.d. random variables.

Existing analyses of HMC \cite{BPRSS2010,BPSSS2011} has focused on how to tune $h$ for optimal efficiency; the expected acceptance probability should be $0.651$.  We also determine how efficiency depends on the mass matrix $V$ and the integration time $T$.  From Theorem \ref{thm hmc conv} we see that $V$ has a similar role to a preconditioner for linear systems of equations.  We should choose $V$ to minimize $\tau$, which is a weighted $4$-norm, while ensuring that the action of matrix multiplication with $V$ and sampling from $\normal(0,V^{-1})$ are cheap to compute.  This result is touched on in \cite{BPSSS2011} where they suggest taking $V = A^{-1}$, the perfect preconditioner, which then requires sampling from $\normal(0,A)$ that can be computationally infeasible.

For Langevin proposals in Section \ref{sec langevin}, $\tau$ was a weighted $6$-norm, see Theorem \ref{thm genlang conv}, but for HMC $\tau$ is a weighted $4$-norm.  As for the Langevin proposal case, this determines the relative importance of how large and small $\lambda_i$ contribute to efficiency, and how $V$ should be chosen.  
HMC is therefore less affected by variation in $\{ \lambda_i \}$ than Langevin proposals due to the $4$-norm instead of $6$-norm dependence of $\tau$.  While the $6$-norm result for Langevin proposals has been seen before in special cases \cite{RR2001,BRS2009}, the appearance of a $4$-norm of $\{\lambda_i \}$ in HMC is new.

Theorem \ref{thm hmc conv} also shows how to choose $T$ to maximize efficiency.  After tuning $h$ to achieve an acceptance probability of $0.651$ the expected squared jump size satisfies
$$
	\mathrm{E}[ |q_i^T V^{-1/2} (x' - x)|^2 ] \rightarrow 1.302 \frac{1-\cos(\lambda_i T')}{\lambda_i^2} \qquad \mbox{as $d \rightarrow \infty$}.
$$
After specifying which $i$ correspond to directions $q_i$ that are important, one we can then choose $T$ to maximize $1 - \cos(\lambda_i T')$ for those $i$.  Thus, an optimal choice of $T$ depends on the target, $V$ and directions of interest.  In the special case when $\lambda_i$ are all equal then $T = \frac{\pi}{\lambda_i}$ is optimal.

The analysis presented here is for HMC with the leap-frog or Stormer-Verlet integrator.  Versions of HMC using higher order integrators are also suggested in \cite{BPRSS2010} and integrators based on splitting methods (in the ODEs context) are suggested in \cite{BCSS2014} that minimize the Hamiltonian error after $L$ steps of the integrator.  It may be possible to evaluate these other methods by first expressing them as stochastic AR(1) proposals and applying this theory.  

The dimension independent version of HMC presented in \cite{BPSSS2011} corresponding to $V = A^{-1}$ and an additional coordinate change of $v = A^{-1} p$ is not analyzed here.  It exactly integrates the Hamiltonian system when the target is Gaussian, so acceptance is guaranteed.

\section{Discussion}
\label{sec discussion}

Until now, each MH algorithm with a stochastic AR(1) proprosal has required separate analysis, e.g. RWM, MALA, SLA, pCN and HMC.  In this article we have designed a unifying theory that encompasses all stationary convergent stochastic AR(1) proposals (not RWM) where $G$ and $\Sigma$ are functions of $A$, for the case where the target distribution is a change of measure from a Gaussian reference measure.  We are no longer constrained to constructing proposals by discretizing differential equations, and are free to consider a wider class of MH algorithms with stochastic AR(1) proposals.  For example, it is now possible to analyze multi-step proposals.  

When the target is $\normal(A^{-1}b,A^{-1})$, then statistical efficiency (expected squared jump size) of these MH algorithms depends on the eigenvalues of the iteration matrix $G$, and the difference between the proposal equilibrium distribution $\normal(\mathcal{A}^{-1}\beta,\mathcal{A}^{-1})$ and the target, see Theorems \ref{thm accept} and \ref{thm jumpsize}.  In particular, if $\normal(\mathcal{A}^{-1}\beta,\mathcal{A}^{-1})$ is identically $\normal(A^{-1}b,A^{-1})$, then proposals are always accepted and jump size is dimension independent provided $\rho(G)$ is bounded strictly below $1$ uniformly in $d$.  Such algorithms are analyzed and accelerated in \cite{F2013,FP2014,FP2016}.   Our theory shows that dimension independent jump size can also be attained when $\normal(\mathcal{A}^{-1}\beta,\mathcal{A}^{-1})$ is not equal to $\normal(A^{-1}b,A^{-1})$, but there are quite strict bounds on the difference, see Lemma \ref{lem accept}, Corollary \ref{cor:accept} and Lemma \ref{lem jump}.

When the target is a change of measure from a Gaussian reference measure, then provided $|\phi_d|$ is bounded, the performance of a MH algorithm with stochastic AR(1) proposal is similar to the normal target case, see Theorems \ref{thm:ng12} and \ref{thm:ng3}.  Corollary \ref{cor:dimindep} provides a sufficient condition for a dimension independent algorithm.  Under weaker conditions on $\phi_d$ (bounded below, a type of Lipschitz continuity, and bounded growth), but stronger conditions on the proposal that effectively force $y \rightarrow x$ as $d \rightarrow \infty$ (see conditions in Theorem \ref{thm:ng1}), then we precisely quantify the acceptance rate and jump size, see Theorems \ref{thm:ng1} and \ref{thm:ng2}.  These final results can be applied to Langevin diffusion-based proposals and HMC.  

Our analysis is different from earlier analyses that relied on a limiting differential equation (Langevin diffusion or Hamiltonian dynamics).  In particular, we have not specified anywhere that the Gaussian reference measure should be of trace class, $\sum_{i} \lambda_i^{-2} < \infty$, since we never define the limiting target distribution as $d \rightarrow \infty$.  Equivalently, Theorems \ref{thm genlang conv}, \ref{thm:lstep} and \ref{thm hmc conv} are not restricted to $\kappa > 1/2$, however, Theorems \ref{thm genlang conv} and \ref{thm:lstep} require that $\lim_{d\rightarrow\infty}\sum_{i=1}^d \lambda_i^{4r-2} < \infty$ for some $r \in \mathbb{R}$ related to the regularity of $\phi_d$.  

Of course, one must also consider the computational cost to determine the true efficiency of an algorithm.  Often, the computational costs are problem specific, so we have been careful to include a range of computationally feasible algorithms in our analysis.  Restricting $G$ and $\Sigma$ to be functions of $A$, as we have done, is a natural restriction to make in high dimensions for computational feasibility.  It also allows a simple coordinate transformation to diagonalize both the proposal and the Gaussian part of the target, reducing the analysis to the case when all matrices are diagonal.  Immediately, this extends the analysis of existing MH algorithms with stochastic AR(1) proposals, such as SLA, pCN and HMC, to target distributions where the covariance of the Gaussian part can have off-diagonal terms.  Computation of the spectral decomposition of $A$ (that defines a coordinate transformation) is not required to implement the methods we analyze.

The performance of initial versions of MALA and HMC was unsatisfactory in high dimensions because they require $\mathcal{O}(N^{1/3})$ and $\mathcal{O}(N^{1/4})$ steps, respectively, to traverse state space.  This led to dimension independent versions of these algorithms, e.g. CN and pCN \cite{CRSW2013}, and a modified version of HMC \cite{BPSSS2011}.  However, these dimension independent methods require the computation of a spectral decomposition of $A$, the action of $A^{-1}$, or sampling from $N(0,A)$ or $N(0,A^{-1})$, which could be computationally infeasible in high dimensions.  Theorems \ref{thm genlang conv} and \ref{thm hmc conv} allow us to use an imperfect preconditioner $V \neq A^{-1}$ for Langevin and HMC-type methods and trade off the benefits of choosing $V \approx A^{-1}$ with the added computational cost this entails.  A `good' preconditioner for Langevin proposals reduces a scaled $6$-norm of the eigenvalues of $VA$, while maintaining cheap computation of independent samples from $\normal(0,V)$ and multiplication with $V$.  For HMC with a normal target, one should choose $V$ to minimize a weighted scaled $4$-norm of the eigenvalues of $VA$, while maintaining cheap computation of samples from $N(0,V)$ and multiplication with $V$.  

Beskos \emph{et al.} \cite{BPRSS2010} showed that $0.651$ is the optimal acceptance probability for HMC when the target distribution has i.i.d. components.  This result also applies for a normal target when the eigenvalues of the preconditioned precision matrix $VA$ grow like $i^{2\kappa}$ for some $\kappa \geq 0$.  After tuning to an acceptance probability of $0.651$, we also provide a condition for choosing the optimal integration time $T$.

Multi-step proposals, that compute $L$ iterations of a stochastic AR(1) proposal before the MH accept/reject step, are a potential `free lunch' if they improve the efficiency of a `single-step' algorithm.  For SLA it is optimal to take $L > 3$ if the dominant computational costs are multiplying with $A$ and evaluating $\phi_d$.  The optimal acceptance probability of multi-step SLA is $0.574$, the same as $1$-step SLA.

We have not analyzed non-stationary stochastic AR(1) proposals where $G$ and $\Sigma$ may depend on $x$, but formally our analysis can be viewed as describing the local behaviour of such algorithms, where the current proposal has proposal equilibrium that is a local Gaussian approximation to the target.  

\appendix
\section{Proofs of Theorems}

\subsection{Proof of Theorem \protect\ref{thm accept}}

We will use the following Lyapunov central limit theorem, see e.g. \cite[Thm. 27.3]{billingsley1995}.

\begin{theorem}
\label{thm clt}
For each $d \in \mathbb{N}$ let $X_{d,1},\dotsc,X_{d,d}$ be a sequence of independent random variables each with finite expected value $\mu_{d,i}$ and variance $\sigma_{d,i}^2$.  Define
$
	s_d^2 := \sum_{i=1}^d \sigma_{d,i}^2.
$
If there exists a $\delta > 0$ such that
$$
	\lim_{d \rightarrow \infty} \frac{1}{s_d^{2+\delta}} \sum_{i=1}^d \mathrm{E}[ |X_{d,i} - \mu_{d,i}|^{2+\delta} ] = 0,
$$
then 
$$
	\frac{1}{s_d} \sum_{i=1}^d (X_{d,i} - \mu_{d,i}) \xrightarrow{\mathcal{D}} \normal(0,1) \qquad \mbox{as $d \rightarrow \infty$}.
$$
\end{theorem}

An equivalent conclusion to this theorem is $\sum_{i=1}^d X_{d,i} \rightarrow \normal(\sum_{i=1}^d \mu_{d,i}, s_d^2)$ in distribution as $d \rightarrow \infty$.  Another useful fact is
\begin{equation}
\label{eq uf1}
	X \sim \normal(\mu,\sigma^2) \qquad \Rightarrow \qquad \mathrm{E}[ 1 \wedge \mathrm{e}^X ] = \Phi(\tsfrac{\mu}{\sigma}) + \mathrm{e}^{\mu + \sigma^2/2} \Phi(-\sigma - \tsfrac{\mu}{\sigma}),
\end{equation}
see e.g. \cite[Prop. 2.4]{RGG1997} or \cite[Lem. B.2]{BRS2009}.

\begin{proof}[Proof of Theorem \ref{thm accept}]
By Lemma \ref{lem2} it is sufficient to only consider \eqref{MH0} in the case where all matrices are diagonal matrices, e.g. $A = \operatorname{diag}(\lambda_1^2,\dotsc,\lambda_d^2)$, $\mathcal{A} = \operatorname{diag}(\tilde{\lambda}_1^2,\dotsc,\tilde{\lambda}_d^2)$, $G = \operatorname{diag}(G_1,\dotsc,G_d)$, $\mu_i = \lambda_i^{-2} b_i$, and $\tilde{\mu}_i = \tilde{\lambda}_i^{-2} \beta_i$.  

Since $\normal(\mathcal{A}^{-1}\beta,\mathcal{A}^{-1})$ is the equilibrium distribution of the proposal chain, $\mathcal{A}^{-1}\beta = G \mathcal{A}^{-1}\beta + g$, so $\tilde{\mu}_i = G_i \tilde{\mu}_i + g_i$.  Also, in equilibrium we have $x_i = \mu_i + \frac{1}{\lambda_i} \xi_i$ where $\xi_i \stackrel{\text{i.i.d.}}{\sim} \normal(0,1)$.  It then follows from \eqref{eq:ar1} and Corollary \ref{lem equiv} that
\begin{align*}
	y_i &= G_i \left( \mu_i + \frac{1}{\lambda_i} \xi_i \right) + g_i  + \frac{(1-G_i^2)^{1/2}}{\tilde{\lambda}_i} \nu_i \\
	&= \tilde{\mu}_i + G_i \hat{r} + \frac{G_i}{\lambda_i} \xi_i + \frac{\hat{g}_i^{1/2}}{\tilde{\lambda}_i} \nu_i \qquad \mbox{where $\nu_i \stackrel{\text{i.i.d.}}{\sim} \normal(0,1)$.}
\end{align*}
From Lemma \ref{lem1} we have $Z = \sum_{i=1}^d Z_{d,i}$ where
$$
	Z_{d,i} = -\tsfrac{1}{2} (\lambda_i^2 - \tilde{\lambda}_i^2)(y_i^2 - x_i^2) + (b_i-\beta_i)(y_i - x_i).
$$
Substituting $x_i$ and $y_i$ as above, using the identity $(b_i - \beta_i)\lambda_i^{-2} = \hat{r}_i + r_i \tilde{\mu}_i$, then after some algebra we eventually find
\begin{align*}
	Z_{d,i} &= T_{0i} + T_{1i}\xi_i + T_{2i}\nu_i + T_{3i} \xi_i^2 + T_{4i} \nu_i^2 + T_{5i} \xi_i \nu_i.
\end{align*}
Hence
$$
	\mu_{d,i}:= \mathrm{E}[Z_{d,i}] = T_{0i} + T_{3i} + T_{4i}
$$
and
\begin{align*}
	\sigma_{d,i}^2 &:= \mathrm{Var}[Z_{d,i}] = \mathrm{E}[Z_{d,i}^2] - \mathrm{E}[Z_{d,i}]^2 \\
	&= \left( T_{0i}^2 + T_{1i}^2 + T_{2i}^2 + 3 T_{3i}^2 + 3 T_{4i}^2 + T_{5i}^2 
	+ 2 T_{0i} T_{3i} + 2 T_{0i} T_{4i} + 2 T_{3i} T_{4i} \right) \\
	& \qquad - \left( T_{0i} + T_{3i} + T_{4i} \right)^2 \\
	&= T_{1i}^2 + T_{2i}^2 + 2 T_{3i}^2 + 2 T_{4i}^2 + T_{5i}^2	 
\end{align*}
and
\begin{align*}
	Z_{d,i} - \mu_{d,i} &= T_{1i} \xi_i + T_{2i} \nu_i + T_{3i} (\xi_i^2 - 1) + T_{4i} (\nu_i^2 - 1) + T_{5i} \xi_i \nu_i.
\end{align*}
Therefore, for any $d \in \mathbb{N}$ and $\delta > 0$ we can bound the Lyapunov condition in Theorem \ref{thm clt} as follows 
\begin{align*}
	\frac{1}{s_d^{2+\delta}} \sum_{i=1}^d \mathrm{E}[|Z_{d,i}-\mu_{d,i}|^{2+\delta}] 
	&\leq \frac{5^{2+\delta}}{s_d^{2+\delta}} \sum_{j=1}^5 C_j(\delta) \sum_{i=1}^d |T_{ji}|^{2+\delta} \\
	&\leq 5^{2+\delta} \sum_{j=1}^5 C_j(\delta) \frac{\sum_{i=1}^d |T_{ji}|^{2+\delta}}{\left( \sum_{i=1}^d T_{ji}^2 \right)^{1+\delta/2} }
\end{align*}
where $C_1(\delta) = C_2(\delta) = \mathrm{E}[|\xi|^{2+\delta}]$ , $C_3(\delta) = C_4(\delta) = \mathrm{E}[|\xi^2-1|^{2+\delta}]$
and $C_5(\delta) = \mathrm{E}[|\xi|^{2+\delta}]^2$, and $\xi \sim \normal(0,1)$.

Therefore, if \eqref{eq Tcond} holds then the result follows from Theorem \ref{thm clt} and \eqref{eq uf1}.
\end{proof}

\subsection{Proof of Theorem \protect\ref{thm jumpsize}}

We need the following technical lemma.

\begin{lemma}
\label{lem 5}
Suppose $\{ t_i \}_{i=1}^\infty \subset \bbR$ and $r > 0$.  Then, for any $k \in \mathbb{N}$,
\begin{equation}
\label{eq l1a}
	\lim_{d\rightarrow \infty} \frac{\sum_{i=1}^d |t_i|^r}{\left( \sum_{i=1}^d t_i^2 \right)^{r/2}} = 0 
	\quad \Rightarrow \quad
	\lim_{d\rightarrow \infty} \frac{\sum_{i=1, i \neq k}^d |t_i|^r}{\left( \sum_{i=1, i \neq k}^d t_i^2 \right)^{r/2}} = 0.
\end{equation}
\end{lemma}

\begin{proof}
Suppose $\lim_{d\rightarrow \infty} (\sum_{i=1}^d |t_i|^r)/ (\sum_{i=1}^d t_i^2 )^{r/2} = 0$ and fix $\rho \in (0,1)$.  Then there exists $D \geq k$ such that for any $d >D$, $|t_k|^r \leq \sum_{i=1}^d |t_i|^r < \rho^{r/2} ( \sum_{i=1}^d t_i^2 )^{r/2}$.  Hence $t_k^2 \leq \rho \sum_{i=1}^d t_i^2$ and so
$$
	\frac{\sum_{i=1, i \neq k}^d |t_i|^r}{\left( \sum_{i=1, i \neq k}^d t_i^2 \right)^{r/2}}  
	< \frac{ \sum_{i=1}^d |t_i|^r }{\left((1-\rho) \sum_{i=1}^d t_i^2 \right)^{r/2}} 
	= \left( \frac{1}{1-\rho} \right)^{r/2} \frac{ \sum_{i=1}^d |t_i|^r }{\left(\sum_{i=1}^d t_i^2 \right)^{r/2}}.
$$
Hence result.
\end{proof}

\begin{proof}[Proof of Theorem \ref{thm jumpsize}]
Under the coordinate transformation $\hat{x} = Q^T x$, \eqref{MH0} becomes \eqref{MH1} and $\mathrm{E}[(q_i^T(x'-x))^2]$ becomes $\mathrm{E}[(\hat{x}_i' - \hat{x}_i)^2]$.  Therefore it is sufficient to only consider $\mathrm{E}[(x_i' - x_i)^2]$ for the case when all matrices are diagonal matrices.  Let $A = \operatorname{diag}(\lambda_1^2,\dotsc,\lambda_d^2)$, $\mathcal{A} = \operatorname{diag}(\tilde{\lambda}_1^2,\dotsc,\tilde{\lambda}_d^2)$, $G = \operatorname{diag}(G_1,\dotsc,G_d)$, $\mu_i = \lambda_i^{-2} b_i$, and $\tilde{\mu}_i = \tilde{\lambda}_i^{-2} \beta_i$.  Since the chain is in equilibrium we have $x_i = \mu_i + \lambda_i^{-1} \xi_i$ for i.i.d. $\xi_i \sim \normal(0,1)$ and $y_i = \tilde{\mu}_i + G_i \hat{r} + G_i \lambda_i^{-1} \xi_i + \hat{g}_i^{1/2} \tilde{\lambda}_i^{-1} \nu_i$ for i.i.d. $\nu_i \sim \normal(0,1)$, see proof of Theorem \ref{thm accept}.  Define $\alpha^-(x,y) := 1 \wedge \exp( \sum_{j=1,j\neq i}^d Z_{d,j} )$ where $Z_{d,j}$ is defined as in the proof of Theorem \ref{thm accept}.  

The proof strategy is to approximate $\mathrm{E}[(x_i' - x_i)^2] = \mathrm{E}[(y_i - x_i)^2\alpha(x,y)]$ with $\mathrm{E}[(y_i-x_i)^2 \alpha^-(x,y)]$;
$$
	\mathrm{E}[(x_i'-x_i)^2] = \mathrm{E}[(y_i-x_i)^2 \alpha^-(x,y)] + \mathrm{E}[(\alpha(x,y)-\alpha^-(x,y))(y_i-x_i)^2].
$$
By independence, 
\begin{align*}
	\mathrm{E}[(y_i-x_i)^2 \alpha^-(x,y)] 
	&= \mathrm{E}[(y_i-x_i)^2] \mathrm{E}[\alpha^-(x,y)] \\
	&= \mathrm{E}\left[ \left(-\tilde{g}_i \hat{r}_i - \frac{\tilde{g}_i}{\lambda_i}\xi_i + \frac{\hat{g}_i^{1/2}}{\tilde{\lambda}_i} \nu_i \right)^2 \right] \mathrm{E}[\alpha^-(x,y)]\\
	&= U_1 \mathrm{E}[\alpha^-(x,y)].
\end{align*}
Also, by Theorem \ref{thm accept} (using Lemma \ref{lem 5} to ensure \eqref{eq Tcond} is met) we obtain $\mathrm{E}[\alpha^-(x,y)] \rightarrow U_2$ as $d \rightarrow \infty$, so $\mathrm{E}[(y_i-x_i)^2 \alpha^-(x,y)] = U_1 U_2 + \mathrm{o}(U_1)$ as $d\rightarrow \infty$.

The error is bounded using the Cauchy-Schwarz inequality;
$$
	|\mathrm{E}[ (\alpha(x,y) - \alpha^-(x,y))(y_i - x_i)^2]|
	\leq \mathrm{E}[(\alpha(x,y)-\alpha^-(x,y))^2]^{1/2} \mathrm{E}[(y_i-x_i)^4]^{1/2}.
$$
Since $1\wedge \mathrm{e}^X$ is Lipschitz with constant $1$, and using results from the proof of Theorem \ref{thm accept}, we obtain
$$
	\mathrm{E}[(\alpha(x,y)-\alpha^-(x,y))^2]^{1/2} 
	\leq \mathrm{E}[Z_{d,i}^2]^{1/2} 
	= ( \sigma_{d,i}^2 + \mu_{d,i}^2 )^{1/2},
$$
and some algebra yields
\begin{align*}
	\mathrm{E}[(y_i-x_i)^4]^{1/2}
	&= \mathrm{E}\left[ \left( -\tilde{g}_i \hat{r}_i - \frac{\tilde{g}_i}{\lambda_i}\xi_i + \frac{\hat{g}_i^{1/2}}{\tilde{\lambda}_i} \nu_i \right)^4 \right]^{1/2} \\
	&= \left( \tilde{g}_i^4 \hat{r}_i^4 +  \frac{3}{\lambda_i^4} (\tilde{g}_i^2 + \tilde{r}_i \hat{g}_i)^2 + \frac{6}{\lambda_i^2} \hat{r}_i^2 \tilde{g}_i^2 ( \tilde{g}_i^2 + \tilde{r}_i \hat{g}_i) \right)^{1/2}.
\end{align*}
\end{proof}

\subsection{Proof of Theorem \protect\ref{thm:ng1} and Corollary \protect\ref{cor:ng}}

First, a lemma which is similar to part of the proof of \cite[Thm. 3]{BRS2009}.

\begin{lemma}
\label{lem:phi}
Suppose $\phi_d$ satisfies Assumption \ref{assumption1} and $r = \max \{ s,s',s'' \}$ satisfies Assumption \ref{assumption2}.  Also suppose that there exists $\{ t_{d,i} \}$ and a $t > 0$ such that $t_{d,i}=\mathcal{O}(d^{-t})$ (uniformly in $i$) as $d \rightarrow \infty$ such that MH algorithm \eqref{MH0ng} satisfies
\begin{align*}
	&\mbox{$G$ and $\Sigma$ are functions of $A$}, \\
	&\tilde{g}_i^2\hat{r}_i^2 \lambda_i^2, \; \tilde{g}_i^2, \; \hat{g}_i \mbox{ are $\mathcal{O}(t_{d,i})$ (uniformly in $i$) as $d \rightarrow \infty$, and } \\
	&\tilde{r}_i \mbox{ is bounded uniformly in $d$ and $i$}.
\end{align*}
Then for any $q' \in \mathbb{N}$ there exists a constant $C > 0$ (that may depend on $q'$) such that
\begin{align*}
	&\mathrm{E}_{\tilde{\pi}_d}[|x-A^{-1}b|_r^{2q'}] < C \mbox{ for all $d$} \\
	&\mathrm{E}_{\tilde{\pi}_d}[|y-x|_r^{2q'}] = \mathcal{O}(t_{d,i}^{q'}) \mbox{ (uniformly in $i$) as $d \rightarrow \infty$},
\end{align*}
and for proposal $y$ from $x$,
$$
	\phi_d(x) - \phi_d(y) \rightarrow 0 \qquad \mbox{in $L^{q'}(\pi_d)$ as $d \rightarrow \infty$}.
$$
\end{lemma}

\begin{proof}
For $x \sim \normal(A^{-1}b,A^{-1})$, $\xi \sim \normal(0,I)$, $\nu \sim \normal(0,I)$ and $\Lambda = \operatorname{diag}(\lambda_1^2,\dotsc,\lambda_d^2)$, 
\begin{align*}
	\mathrm{E}_{\tilde{\pi}_d}[ |x - A^{-1}b |_r^{2q'} ] 
	&= \mathrm{E}_{\tilde{\pi}_d}[ | A^{r-1/2} \xi |^{2q'} ] 
	= \mathrm{E}_{\tilde{\pi}_d}[ | \Lambda^{r-1/2} \nu |^{2q'} ] \\
	&= \mathrm{E}_{\tilde{\pi}_d}\left[ \left( \sum_{i=1}^d \lambda_i^{4r-2} \nu_i^2 \right)^{q'} \right] 
	\leq C \left( \sum_{i=1}^d \lambda_i^{4r-2} \right)^{q'},
\end{align*}
which is bounded uniformly in $d$ by Assumption \ref{assumption2}.

As above, and using the transformation $\hat{x} = Q^T x$, and $(y_i - x_i) = -\tilde{g}_i \hat{r}_i - \tilde{g}_i \lambda_i^{-1} \xi_i + \hat{g}_i^{1/2} \tilde{\lambda}_i^{-1} \nu_i$ from the proof of Theorem \ref{thm jumpsize} where $\xi_i$ and $\nu_i \iid \normal(0,1)$, we have
\begin{align*}
	\mathrm{E}_{\tilde{\pi}_d}[ |y - x|_r^{2q'} ]
	&= \mathrm{E}_{\tilde{\pi}_d}\left[ \left( \sum_{i=1}^d \lambda_i^{4r} \left( -\tilde{g}_i \hat{r}_i - \frac{\tilde{g}_i}{\lambda_i} \xi_i + \frac{\hat{g}_i^{1/2}}{\tilde{\lambda}_i} \nu_i \right)^2 \right)^{q'} \right]  \\
	&= \mathrm{E}_{\tilde{\pi}_d}\left[ \left( \sum_{i=1}^d \lambda_i^{4r-2} \left( -\tilde{g}_i \hat{r}_i \lambda_i - \tilde{g}_i \xi_i + \tilde{r}_i^{1/2} \hat{g}_i^{1/2}  \nu_i \right)^2 \right)^{q'} \right]  .
\end{align*}
Since $\tilde{g}_i^2 \hat{r}_i^2 \lambda_i^2$, $\tilde{g}_i^2$, $\hat{g}_i = \mathcal{O}(t_{d,i})$, and $\tilde{r}_i$ is bounded uniformly in $d$ and $i$, it follows that for all sufficiently large $d$, 
$$
	\mathrm{E}_{\tilde{\pi}_d}[ |y - x|_r^{2q'} ] \leq C t_{d,i}^{q'} \left( \sum_{i=1}^d \lambda^{4r-2} \right)^{q'},
$$
so by Assumption \ref{assumption2}, $\mathrm{E}_{\tilde{\pi}_d}[ |y - x|_r^{2q'} ] = \mathcal{O}(t_{d,i}^{q'})$ as $d \rightarrow \infty$.

From $\mathrm{E}_{\tilde{\pi}_d}[|x-A^{-1}b|_r^{2q'}] < C$ and $\mathrm{E}_{\tilde{\pi}_d}[|y-x|_r^{2q'}] \rightarrow 0$, it follows from the triangle inequality that there is a (new) constant $C>0$ such that
\begin{equation}
\label{eq:ybd}
	\mathrm{E}_{\tilde{\pi}_d}[|y-A^{-1}b|_r^{2q'}] < C \qquad \mbox{for all $d$}.
\end{equation}

Let $\Delta_d = \phi_d(x) - \phi_d(y)$.  For any $R>0$ and $q' \in \mathbb{N}$ define
\begin{align*}
	\gamma(R) &= \sup \{  \delta(a,b)^{q'} : a\leq R,b\leq R \} \\
	S_1 &= \{ x \in \bbR^d : |x-A^{-1}b|_s \leq R \}, \\
	S_2 &= \{ x \in \bbR^d: |y-A^{-1}b|_s \leq R \},
\end{align*}
and let $\mathbb{I}_{S}$ be the indicator function for set $S$.  Using Assumption \ref{assumption1}, a generic constant $C$ that may vary between lines, the Cauchy-Schwarz inequality, and then Markov's inequality, we have for each $q' \in \mathbb{N}$
\begin{align*}
	\mathrm{E}_{\tilde{\pi}_d}[|\Delta_d|^{q'}] 
	 &= \mathrm{E}_{\tilde{\pi}_d}[|\Delta_d|^{q'} \mathbb{I}_{S_1 \cap S_2} ] +  \mathrm{E}_{\tilde{\pi}_d}[|\Delta_d|^{q'} \mathbb{I}_{\bbR^d \backslash (S_1 \cap S_2)}] \\
	&\leq \gamma(R) \mathrm{E}[ |x-y|_{s'}^{q'} ] + C \mathrm{E}[(1+|x-A^{-1}b|_{s''}^{pq'} + |y-A^{-1}b|_{s''}^{pq'}) \mathbb{I}_{\bbR^d \backslash (S_1 \cap S_2)}] \\
	&\leq \gamma(R) \mathrm{E}[ |x-y|_{s'}^{q'} ] \\
	&\qquad + C \mathrm{E}[1+|x-A^{-1}b|_{s''}^{2pq'} + |y-A^{-1}b|_{s''}^{2pq'}]^{1/2} (\mathrm{P}(\bbR^d\backslash S_1) + \mathrm{P}(\bbR^d \backslash S_2))^{1/2} \\
	&\leq \gamma(R) \mathrm{E}[ |x-y|_{s'}^{q'} ] + C (\mathrm{P}(\bbR^d\backslash S_1) + \mathrm{P}(\bbR^d \backslash S_2))^{1/2} \\
	&\leq \gamma(R) \mathrm{E}[ |x-y|_{s'}^{q'} ] + \frac{C}{R^{1/2}}  \left( \mathrm{E}[|x-A^{-1}b|_s] + \mathrm{E}[|y-A^{-1}b|_s] \right)^{1/2}\\
	&\leq \gamma(R) \mathrm{E}[ |x-y|_{s'}^{q'} ] + \frac{C}{R^{1/2}}.
\end{align*}
Note that we used Jensen's inquality (which implies $\nrm{f}_{L^a(\tilde{\pi}_d)} \leq \nrm{f}_{L^b(\tilde{\pi}_d)}$ for $1\leq a \leq b \leq \infty$) and \eqref{eq:normbound} to obtain bounds on $\mathrm{E}[|x-A^{-1}b|_{s''}^{2pq'}]$, $\mathrm{E}[|y-A^{-1}b|_{s''}^{2pq'}]$, $\mathrm{E}[|x-A^{-1}b|_s]$ and $\mathrm{E}[|y-A^{-1}b|_s]$.

Hence, for any $\epsilon > 0$ we can choose $R = R(\epsilon)$ such that $C/R^{1/2} < \epsilon/2$ and since $\mathrm{E}_{\tilde{\pi}_d}[|y-x|_r^{q'}] \rightarrow 0$ as $d \rightarrow \infty$ (by Jensen's inequality), for all sufficiently large $d$ we have
$$ 
\mathrm{E}_{\tilde{\pi}_d}[|\phi_d(x) - \phi_d(y)|] = \mathrm{E}_{\tilde{\pi}_d}[|\Delta_d|^{q'}] < \epsilon.
$$
Thus, 
$$
	\phi_d(x)-\phi_d(y) \rightarrow 0 \qquad \mbox{ in $L^{q'}(\tilde{\pi}_d)$ as $d \rightarrow \infty$}.
$$
The result then follows from $\phi_d(x) \geq m$.
\end{proof}

\begin{proof}[Proof of Theorem \ref{thm:ng1}]
As in the proofs of Theorems \ref{thm accept} and \ref{thm jumpsize} it is sufficient to prove the result in the case when all matrices are diagonal.  This follows from the coordinate transformation $\hat{x} = Q^T x$ where $A = Q \Lambda Q^T$, since
$$
	|x|_s = |\Lambda^s \hat{x}| \qquad \mbox{and} \qquad
	|x - A^{-1}b|_s = |\Lambda^s( \hat{x} - \Lambda^{-1}Q^Tb)|
$$
for all $x \in \bbR^d$ and $s \in \bbR$, and since $\psi_d(\hat{x}) := \phi_d(Q \hat{x})$ satisfies Assumption \ref{assumption1}.  Henceforth and without loss of generality, let us assume that $A$, $G$ and $\Sigma$ are diagonal matrices.

Using Lemma \ref{lem:phi} with $q' = 1$, and the fact that $z \mapsto 1 \wedge \exp(z)$ is globally Lipschitz continuous, it follows that
\begin{equation}
\label{eq:acceptlim}
	\mathrm{E}_{\pi_d}[\alpha(x,y)] - \mathrm{E}_{\pi_d}[\tilde{\alpha}(x,y)] \rightarrow 0 \qquad \mbox{as $d \rightarrow \infty$}.
\end{equation}
To complete the proof we must find the limit of $\mathrm{E}_{\pi_d}[\tilde{\alpha}(x,y)]$ as $d \rightarrow \infty$.

As in the proof of Theorem \ref{thm accept} we have $Z = \sum_{i=1}^d Z_{d,i}$ where
$$
	Z_{d,i} = T_{0i} + T_{1i} \xi_i + T_{2i} \nu_i + T_{3i} \xi_i^2 + T_{4i} \nu_i^2 + T_{5i} \xi_i \nu_i,
$$
noting that $\xi_i = \lambda_i (x_i - \mu_i)$ with $x \sim \pi_d$, $\mu_i = (A^{-1}b)_i$, and $\nu_i \stackrel{\text{i.i.d.}}{\sim} \normal(0,1)$.  

Note that $\kappa_{d,i} = \mathrm{E}_{\pi_d}[ \lambda_i (x_i - \mu_i)] = \mathrm{E}_{\pi_d}[\xi_i]$ and $\gamma_{d,i} = \mathrm{E}_{\pi_d}[ \lambda_i^2 (x_i - \mu_i)^2] = \mathrm{E}_{\pi_d}[\xi_i^2]$ are both uniformly bounded in $d$ and $i$ since $\phi_d(x) \geq m$ and $|\kappa_{d,i}|\leq \mathrm{e}^{-m} \mathrm{E}_{\tilde{\pi}_d}[|\xi_i|] = \mathrm{e}^{-m} \mathrm{E}[|u|]$ and $0 \leq \gamma_{d,i} \leq \mathrm{e}^{-m} \mathrm{E}_{\tilde{\pi}_d}[\xi_i^2] = \mathrm{e}^{-m} \mathrm{E}[u^2]$ where $u \sim \normal(0,1)$.  If we define 
$$
	S_{d,j} := \sum_{i=1}^j T_{1i} (\xi_i - \kappa_{d,i}) + T_{2i} \nu_i + T_{3i} (\xi_i^2 - \gamma_{d,i}) + T_{4i} (\nu_i^2 - 1) + T_{5i} \xi_i \nu_i,
$$
then
$$
	Z = \sum_{i=1}^d (T_{0i} + T_{1i} \kappa_{d,i} + T_{3i} \gamma_{d,i} + T_{4i}) + S_{d,d}.
$$

We will now show that $S_{d,d}$ converges in distribution towards a normal distribution as $d \rightarrow \infty$, using a Martingale central limit theorem, see \cite[Thm. 3.2, p. 58]{HallHeyde}.  

The set $\{ S_{d,j} : 1 \leq j \leq d, d \in \mathbb{N} \}$ is a zero mean, square-integrable Martingale array, i.e. for each $d \in \mathbb{N}$ and $1 \leq j \leq d$, $S_{d,j}$ is measurable, 
$$
	\mathrm{E}_{\pi_d}[S_{d,j}] = 0, \qquad
	\mathrm{E}_{\pi_d}[|S_{d,j}|]<\infty, \qquad \mbox{and} \qquad
	\mathrm{E}_{\pi_d}[(S_{j,d})^2] < \infty.
$$
For definitions, see \cite[p. 1 and 53]{HallHeyde}.  Define $X_{d,j} := S_{d,j} - S_{d,j-1}$.  To ensure we satisfy the conditions for \cite[Thm. 3.2]{HallHeyde} we must show that there exists an a.s. finite random variable $\eta^2$ such that
\begin{align}
	&\max_{1 \leq i \leq d} | X_{d,i} | \xrightarrow{p} 0 \qquad \mbox{as $d \rightarrow \infty$}, \label{eq:c1}\\
	&\sum_{i=1}^d X_{d,i}^2 \xrightarrow{p} \eta^2 \qquad \mbox{as $d \rightarrow \infty$, and} \label{eq:c2} \\
	&\mathrm{E}_{\pi_d} \left( \max_{1\leq i \leq d} X_{d,i}^2 \right) \mbox{ is bounded in $d$}. \label{eq:c3}
\end{align}

First consider \eqref{eq:c1}.  We have
$$
	|X_{d,i}| \leq |T_{1i}|(|\xi_i| + |\kappa_{d,i}|) + |T_{2i}| |\nu_i| + |T_{3i}| (\xi_i^2 + \gamma_{d,i}) + |T_{4i}| (\nu_i^2 + 1) + |T_{5i}| |\xi_i| |\nu_i|
$$
which goes to zero in probability since $\kappa_{d,i}$ and $\gamma_{d,i}$ are bounded uniformly and $|T_{ji}|$ are all $\mathcal{O}(d^{-1/2})$ as $d \rightarrow \infty$ uniformly in $i$.

Now consider \eqref{eq:c2}.  Define $\eta^2:= \sigma_{ng}^2$ so that
$$
	\eta^2 := \lim_{d \rightarrow \infty} \sum_{i=1}^d T_{1i}^2 + T_{2i}^2 + 2 T_{3i}^2 + 2 T_{4i}^2 + T_{5i}^2 + \left( \kappa_{d,i} T_{1i} + T_{3i} ( \gamma_{d,i} - 1 ) \right)^2 < \infty
$$
and $Y_{d,i} := d X_{d,i}^2$ so that $\overline{Y}_{d} = \frac{1}{d} \sum_{i=1}^d Y_{d,i} = \sum_{i=1}^d X_{d,i}^2$.  Then
$$
	\mathrm{E}_{\tilde{\pi}_d}[Y_{d,i}] = d \left( T_{1i}^2 + T_{2i}^2 + 2 T_{3i}^2 + 2 T_{4i}^2 + T_{5i}^2 + (\kappa_{d,i}T_{1i} + T_{3i} (\gamma_{d,i} - 1))^2 \right),
$$
and
$$
	\mathrm{Var}_{\tilde{\pi}_d}[Y_{d,i}] = d^2 \sum_{|\omega| = 4} C_{\omega} T_{1i}^{\omega_1} T_{2i}^{\omega_2} T_{3i}^{\omega_3} T_{4i}^{\omega_4} T_{5i}^{\omega_5}
$$
where $\omega$ is a multi-index with $\omega_i \geq 0$ and $|\omega| = \sum_{i} \omega_i = 4$, and $C_\omega$ are uniformly bounded constants.  Since $T_{ji}$ are all uniformly $\mathcal{O}(d^{-1/2})$, $\mathrm{E}_{\tilde{\pi}_d}[Y_{d,i}]$ and $\mathrm{Var}_{\tilde{\pi}_d}[Y_i]$ are uniformly bounded in $d$ and $i$.

Then by the Markov inequality, and independence of $Y_{d,i}$, for any $\epsilon > 0$,
\begin{align*}
	\mathrm{Pr}\left( \left| \overline{Y}_d - \mathrm{E}_{\tilde{\pi}_d}\left[ \overline{Y}_d \right] \right| \geq \epsilon \right) 
	&\leq \frac{1}{\epsilon^2} \mathrm{Var}_{\tilde{\pi}_d}\left[ \overline{Y}_d \right] 
	= \frac{1}{\epsilon^2 d^2} \sum_{i=1}^d \mathrm{Var}_{\tilde{\pi}_d} [Y_{d,i} ] \\
	&\leq \frac{C}{\epsilon^2 d} \rightarrow 0 \quad \mbox{as $d \rightarrow \infty$}.
\end{align*}
Hence $\overline{Y}_d \xrightarrow{p} \eta^2$ as $d \rightarrow \infty$.  This is not \eqref{eq:c2} yet, because it is convergence with respect to $\tilde{\pi}_d$ rather than $\pi_d$.  

Since $\lim_{d \rightarrow \infty} \mathrm{E}_{\tilde{\pi}_d}[ \overline{Y}_d ] = \eta^2 < \infty$ and $|\overline{Y}_d| = \overline{Y}_d$, it follows that $\mathrm{E}_{\tilde{\pi}_d}[|\overline{Y}_d| ]$ is uniformly bounded in $d$.  Therefore, $\overline{Y}_d$ is uniformly integrable and so $\overline{Y}_d \rightarrow \eta^2$ in $L^1(\tilde{\pi}_d)$ as $d \rightarrow \infty$ \cite[Thm. 6.5.5 on p. 169]{sensinger1993}.  Hence 
$$
	\sum_{i=1}^d X_{d,i}^2 \xrightarrow{L^1(\tilde{\pi}_d)} \eta^2 \qquad \mbox{as $d \rightarrow \infty$}.
$$
From $\phi_d(x) \geq m$, the same limit holds in $L^1(\pi_d)$, which also implies convergence in probability, hence we have shown \eqref{eq:c2}.

Condition \eqref{eq:c3} follows from $X_{d,i}^2 \leq \overline{Y}_d$ for $1 \leq i \leq d$, $\mathrm{E}_{\tilde{\pi}_d}[|\overline{Y}_d| ]$ uniformly bounded in $d$, and $\phi_d(x) \geq m$.

Therefore, by the Martingale central limit theorem \cite[Thm. 3.2]{HallHeyde}, 
$$
	S_{d,d} \xrightarrow{\mathcal{D}} \normal(0,\eta^2) \qquad \mbox{as $d \rightarrow \infty$}.
$$	
Hence, 
\begin{equation}
\label{eq:zlim}
	Z  \xrightarrow{\mathcal{D}} \normal(\mu_{ng}, \sigma_{ng}^2 ) \qquad \mbox{as $d \rightarrow \infty$}.
\end{equation}
The result then follows from \eqref{eq:acceptlim}, \eqref{eq:zlim} and \eqref{eq uf1}.
\end{proof}

\begin{proof}[Proof of Corollary \ref{cor:ng}]
With $\xi_i$ defined as in the proof of Theorem \ref{thm:ng1} (we only need to consider the case when matrices are diagonal), since $\mathrm{E}_{\tilde{\pi}_d}[\xi_i^2] = 1$ and $\mathrm{E}_{\tilde{\pi}_d}[\xi_i^4] = 3$ for all $i$ and $d$, and since $\xi_i$ and $\xi_j$ are i.i.d. under $\tilde{\pi}_d$,
$$
	\lim_{d \rightarrow \infty} \mathrm{E}_{\tilde{\pi}_d}\left[ \left( \sum_{i=1}^d T_{1i} \xi_i + T_{3i} (\xi_i^2 - 1) \right)^2 \right] = \lim_{d \rightarrow \infty} \sum_{i=1}^d T_{1i}^2 + 2T_{3i}^2 = 0
$$
From Jensen's inequality and $\phi_d(x) \geq m$ we have $\sum_{i=1}^d (T_{1i} \xi_i + T_{3i} (\xi_i^2 - 1)) \rightarrow 0$ in $L^1(\pi_d)$ as $d \rightarrow \infty$.  Therefore, since $\kappa_{d,i} = \mathrm{E}_{\pi_d}[\xi_i]$ and $\gamma_{d,i} = \mathrm{E}_{\pi_d}[\xi_i^2]$,
$$
	\sum_{i=1}^d T_{1i} \kappa_{d,i} + T_{3i}(\gamma_{d,i} - 1) = \mathrm{E}_{\pi_d}\left[ \sum_{i=1}^d T_{1i} \xi_i + T_{3i}(\xi_i^2 - 1) \right] \rightarrow 0 \qquad \mbox{as $d \rightarrow \infty$},
$$
and $\mu_{ng} = \mu$.  Also, since $\kappa_{d,i}$ and $\gamma_{d,i}$ are uniformly bounded in $d$ and $i$,
$$
	\sum_{i=1}^d (T_{1i}\kappa_{d,i} + T_{3i} (\gamma_{d,i}-1))^2 \leq C \sum_{i=1}^d T_{1i}^2 + T_{3i}^2 \rightarrow 0 \qquad \mbox{as $d \rightarrow \infty$},
$$
and $\sigma_{ng}^2 = \sigma^2$.
\end{proof}

\subsection{Proof of Theorem \protect\ref{thm:ng2}}

\begin{proof}
As in earlier proofs, it is sufficient to prove the result in the case when all matrices are diagonal, so let $A$, $G$ and $\Sigma$ be diagonal matrices.  

Let $S_d$ denote the expected squared jump size in coordinate direction $i$, so that
$$
	S_d = \mathrm{E}_{\pi_d}[(x_i' - x_i)^2] = \mathrm{E}_{\pi_d}[(y_i - x_i)^2 \alpha(x,y)].
$$
Also define
\begin{align*}
	\tilde{S}_d := \mathrm{E}_{\pi_d}[(y_i-x_i)^2 \tilde{\alpha}(x,y)] \quad \mbox{and} \quad
	\tilde{S}_d^- := \mathrm{E}_{\pi_d}[(y_i-x_i)^2 \tilde{\alpha}^-(x,y)]
\end{align*}
where $\tilde{\alpha}^-(x,y) = 1 \wedge \exp( \sum_{j=1,j\neq i}^d Z_{d,i})$, and $Z_{d,i}$ is the same as earlier.  Recall that $\xi_i = \lambda_i (x_i-\mu_i)$.

We first show that $\mathrm{E}_{\pi_d}[(y_i-x_i)^4] = \mathcal{O}(t_{d,i}^2 \lambda_i^{-4})$ (uniformly in $i$) as $d \rightarrow \infty$.  As in the proof of Theorem \ref{thm jumpsize}, from $y = Gx + g + \Sigma^{1/2} \nu$ where $\nu \sim \normal(0,I)$, and $\mathcal{A}^{-1}\beta = G \mathcal{A}^{-1}\beta + g$ it follows that
$$
	y_i - x_i = -\tilde{g}_i \hat{r}_i - \frac{\tilde{g}_i}{\lambda_i} \xi_i + \frac{\hat{g}_i^{1/2}}{\tilde{\lambda}_i} \nu_i
$$
where $\nu_i \stackrel{\text{i.i.d.}}{\sim} \normal(0,1)$, $\xi_i = \lambda_i (x_i - \mu_i)$ and $x \sim \pi_d$.  Therefore,
$$
	(y_i - x_i)^4 = \lambda_i^{-4} \left( -\tilde{g}_i \hat{r}_i \lambda_i - \tilde{g}_i \xi_i + \tilde{r}_i^{1/2} \hat{g}_i^{1/2} \nu_i \right)^4.
$$
In the proof of Theorem \ref{thm:ng1} we showed that $|\kappa_{d,i}| = |\mathrm{E}_{\pi_d}[\xi_i]|$ and $\gamma_{d,i} = \mathrm{E}_{\pi_d}[\xi_i^2]$ are uniformly bounded.  Similarly, $|\mathrm{E}_{\pi_d}[\xi_i^3]|$ and $\mathrm{E}_{\pi_d}[\xi_i^4]$ are also uniformly bounded.  Using these facts and $\tilde{g}_i^2 \hat{r}_i^2 \lambda_i^2$, $\tilde{g}_i^2 $, $\hat{g}_i$ $\mathcal{O}(t_{d,i})$ (uniformly in $i$) and $\tilde{r}_i$ bounded uniformly in $d$ and $i$, it follows that $\mathrm{E}_{\pi_d}[(y_i-x_i)^4] = \mathcal{O}(t_{d,i}^2 \lambda_i^{-4})$ (uniformly in $i$) as $d \rightarrow \infty$.

Now let us show that $S_d - \tilde{S}_d = \mathrm{o}(t_{d,i}\lambda_i^{-2})$  as $d \rightarrow \infty$.  From the Lipschitz continuity of $z \mapsto 1 \wedge \exp(z)$ and the Cauchy-Schwarz inequality,
\begin{align*}
	S_d - \tilde{S}_d &= \mathrm{E}_{\pi_d}[ (y_i - x_i)^2 (\alpha(x,y) - \tilde{\alpha}(x,y) ] \\
	&\leq \mathrm{E}_{\pi_d}[(y_i-x_i)^4]^{1/2} \mathrm{E}_{\pi_d}[ |\phi_d(x) - \phi_d(y) |^2 ]^{1/2} = \mathrm{o}(t_{d,i} \lambda_i^{-2}) \quad \mbox{as $d \rightarrow \infty$},
\end{align*}
by Lemma \ref{lem:phi}, since $\mathrm{E}_{\pi_d}[(y_i-x_i)^4]  =\mathcal{O}(t_{d,i}^2\lambda_i^{-4})$.

Now show that $\tilde{S}_d - \tilde{S}_d^- = \mathrm{o}(t_{d,i} \lambda_i^{-2})$ as $d \rightarrow \infty$.  Again, by the Lipschitz continuity of $z \mapsto 1 \wedge \exp(z)$ and the Cauchy-Schwarz inequality,
\begin{align*}
	\tilde{S}_d - \tilde{S}_d^- &= \mathrm{E}_{\pi_d}[ (y_i - x_i)^2 (\tilde{\alpha}(x,y) - \tilde{\alpha}^-(x,y))] \\
	&\leq \mathrm{E}_{\pi_d}[(y_i-x_i)^4]^{1/2} \mathrm{E}_{\pi_d}[ Z_{d,i}^2 ]^{1/2} = \mathrm{o}(t_{d,i}\lambda_i^{-2}) \quad \mbox{as $d \rightarrow \infty$},
\end{align*}	
since $\mathrm{E}_{\pi_d}[(y_i-x_i)^4] = \mathcal{O}(t_{d,i}^2 \lambda_i^{-4})$ and $T_{ji}$ are all $\mathcal{O}(d^{-1/2})$.

Now consider $\tilde{S}_d^-$.  Since $\nu_i$ is independent of $\xi_i$ and $\tilde{\alpha}^-(x,y)$, $\mathrm{E}[\nu_i] = 0$ and $\mathrm{E}[\nu_i^2] = 1$,
\begin{align*}
	\tilde{S}_d^-
	&= \mathrm{E}_{\pi_d}[ (y_i - x_i)^2 \tilde{\alpha}^-(x,y)] \\
	&= \mathrm{E}_{\pi_d}\left[ \left( \tilde{g}_i^2 \hat{r}_i^2 + 2 \frac{\hat{r}_i \tilde{g}_i^2}{\lambda_i} \xi_i + \frac{\tilde{g}_i^2}{\lambda_i^2} \xi_i^2 + \frac{\hat{g}_i}{\tilde{\lambda}_i^2} \right) \tilde{\alpha}^-(x,y) \right] \\
	&= \left( \tilde{g}_i^2 \hat{r}_i^2 + \frac{\hat{g}_i}{\tilde{\lambda}_i^2} \right) \mathrm{E}_{\pi_d}[ \tilde{\alpha}^-(x,y)] + 2 \frac{\hat{r}_i \tilde{g}_i^2}{\lambda_i} \mathrm{E}_{\pi_d} [\xi_i \tilde{\alpha}^-(x,y) ] + \frac{\tilde{g}_i^2}{\lambda_i^2} \mathrm{E}_{\pi_d} [\xi_i^2 \tilde{\alpha}^-(x,y) ].
\end{align*}
Since $\tilde{\alpha}^-(x,y) \in (0,1]$, it follows from Jensen's inequality that 
$$
	| \mathrm{E}_{\pi_d}[\xi_i \tilde{\alpha}^-(x,y)] | \leq \mathrm{E}_{\pi_d}[|\xi_i| \tilde{\alpha}^-(x,y)] \leq \mathrm{E}_{\pi_d}[|\xi_i|] \leq \mathrm{E}_{\pi_d}[\xi_i^2]^{1/2} = \gamma_{d,i}^{1/2}.
$$
Also, $0 \leq \mathrm{E}_{\pi_d} [\xi_i^2 \tilde{\alpha}^-(x,y) ] \leq \mathrm{E}_{\pi_d} [\xi_i^2 ] = \gamma_{d,i}$.  

Finally, using $z \mapsto 1\wedge \exp(z)$ Lipschitz, and since $T_{ji}$ are $\mathcal{O}(d^{-1/2})$,
$$
	|\mathrm{E}_{\pi_d}[ \tilde{\alpha}^-(x,y)] - \mathrm{E}_{\pi_d}[ \tilde{\alpha}(x,y) ]|
	\leq \mathrm{E}_{\pi_d}|Z_{d,i}| \rightarrow 0 \qquad \mbox{as $d \rightarrow \infty$}.
$$
Since $\tilde{g}_i^2 \hat{r}_i^2 + \frac{\hat{g}_i}{\tilde{\lambda}_i^2} = \mathcal{O}(t_{d,i}\lambda_i^{-2})$ it follows that as $d \rightarrow \infty$
\begin{align*}
	S_d &=\tilde{S}_d^- + \mathrm{o}(t_{d,i} \lambda_i^{-2}) \\
	&= \left( \left(\tilde{g}_i^2 \hat{r}_i^2 + \frac{\hat{g}_i}{\tilde{\lambda}_i^2} \right) \mathrm{E}_{\pi_d}[\tilde{\alpha}(x,y)] + 2 \frac{\hat{r}_i \tilde{g}_i^2 \gamma_{d,i}^{1/2}}{\lambda_i} u_{d,i} + \frac{\tilde{g}_i^2 \gamma_{d,i}}{\lambda_i^2} v_{d,i} \right) + \mathrm{o}(t_{d,i} \lambda_i^{-2}) 
\end{align*}
for some  $u_{d,i} \in [-1,1]$ and $v_{d,i} \in [0,1]$.  The result then follows from \eqref{eq:acceptlim} and $\tilde{g}_i^2 \hat{r}_i^2 + \frac{\hat{g}_i}{\tilde{\lambda}_i^2} = \mathcal{O}(t_{d,i}\lambda_i^{-2})$.
\end{proof}

\subsection{Proof of Theorem \ref{thm genlang conv}}

We can use the following technical lemma in the proof of Theorem \ref{thm genlang conv}.

\begin{lemma}
\label{lem:a4}
Suppose $\{ t_i \} \subset \mathbb{R}$ is a sequence satisfying $0 \leq t_i \leq C d^{-1/3} (\frac{i}{d})^{2\kappa}$ for some $C > 0$ and $\kappa \geq 0$.  If $s>3$, then $\lim_{d \rightarrow \infty} \sum_{i=1}^d t_i^s = 0$.
\end{lemma}

\begin{proof}
\begin{align*}
\lim_{d \rightarrow \infty} \sum_{i=1}^d t_i^s \leq C^s \lim_{d\rightarrow \infty} d^{1-s/3} \sum_{i=1}^d \tsfrac{1}{d} (\tsfrac{i}{d})^{2\kappa s} = C^s \lim_{d \rightarrow \infty} d^{1-s/3} \int_0^1 z^{2\kappa s} \dd z = 0.
\end{align*}
\end{proof}

\begin{proof}[Proof of Theorem \ref{thm genlang conv}]
Define $t_{d,i} = h \lambda_i^2 = \mathcal{O}(d^{-1/3})$, $\rho = (\theta-\frac{1}{2})/2$ and $w_{d,i} = 1 + \frac{\theta}{2} t_{d,i}$.  Applying Theorem \ref{thm:2.1} to MH algorithm \eqref{eq hat1} we have $\mathcal{A} = B + \rho h B^2$ and $\mathcal{A}^{-1}\beta = B^{-1} V^{1/2} b$.  Then 
\begin{align*}
	\tilde{\lambda}_i^2 &= (1 + \rho t_{d,i}) \lambda_i^2, &
	G_i &= 1 - \frac{\frac{1}{2}t_{d,i}}{w_{d,i}}, &
	\tilde{g}_i&= \frac{\frac{1}{2}t_{d,i}}{w_{d,i}}, &
	\hat{g}_i &= \frac{t_{d,i} (1 + \rho t_{d,i})}{w_{d,i}^2}, \\
	r_i &= -\rho t_{d,i}, &
	\tilde{r}_i &= \frac{1}{1 + \rho t_{d,i}}, &
	\hat{r}_i &= 0,
\end{align*}
so that $T_{0i} = T_{1i} = T_{2i} = 0$ and
\begin{align*}
	T_{3i} &= - \frac{\rho t_{d,i}^2 (1 + \rho t_{d,i})}{2 w_{d,i}^2} &
	T_{4i} &= \frac{\rho t_{d,i}^2}{2 w_{d,i}^2} &
	T_{5i} &= \frac{\rho t_{d,i}^{3/2}}{w_{d,i}^2} \left( 1 - \tsfrac{(1-\theta)}{2} t_{d,i} \right)
\end{align*}
Then $\tilde{g}_i^2 \hat{r}_i^2 \lambda_i^2 =0$, $\tilde{g}_i^2 = \mathcal{O}(t_{d,i}^2)$, $\hat{g}_i = \mathcal{O}(t_{d,i})$, $\tilde{r}_i = \mathcal{O}(1)$, $T_{3i} = \mathcal{O}(d^{-2/3})$, $T_{4i} = \mathcal{O}(d^{-2/3})$, and $T_{5i} = \mathcal{O}(d^{-1/2})$ (uniformly in $i$) as $d \rightarrow \infty$.

Also, using Lemma \ref{lem:a4}
$$
	\mu = \lim_{d \rightarrow \infty} \sum_{i=1}^d T_{3i} + T_{4i} = - \frac{l^6 (\theta-\frac{1}{2})^2 \tau}{8}
$$
and
$$
	\sigma^2 = \lim_{d \rightarrow \infty} \sum_{i=1}^d 2 T_{3i}^2 + 2 T_{4i}^2 + T_{5i}^2 = - \frac{l^6 (\theta-\frac{1}{2})^2 \tau}{4}.
$$
Hence $\lim_{d\rightarrow\infty} \sum_{i=1}^d T_{1i}^2 + T_{2i}^2 + 2 T_{3i}^2 + 2 T_{4i}^2 + T_{5i}^2 < \infty$.  So the conditions for applying Theorems \ref{thm:ng1} and \ref{thm:ng2} are satisfied.

Moreover, using Lemma \ref{lem:a4}, $ \lim_{d\rightarrow \infty} \sum_{i=1}^d T_{1i}^2 + T_{3i}^2 = 0$, so we can apply Corollary \ref{cor:ng}.

It follows that $\frac{\mu}{\sigma} = -\sigma - \frac{\mu}{\sigma} = -l^3 |\theta-\frac{1}{2}| \sqrt{\tau}/4$ and $\mu + \frac{\sigma^2}{2} = 0$.  Hence we obtain \eqref{eq expect} from Theorem \ref{thm:ng1}, using Corollary \ref{cor:ng}.

For the expected jump size, first note that $\hat{r}_i = 0$.  Also, $\hat{g}_i \tilde{\lambda}_i^{-2}  = h (1 + \tsfrac{\theta}{2} t_{d,i})^{-1} = h + \mathcal{O}(h t_{d,i}) =h + \mathrm{o}(h)$ as $d \rightarrow \infty$, and using the fact that $\gamma_{d,i}^2$ is uniformly bounded (see proof of Theorem \ref{thm:ng1}), $\tilde{g}_i^2 \gamma_{d,i}^{1/2} \lambda_i^{-2} = \mathcal{O}(t_{d,i}^2 \lambda_i^{-2}) = \mathcal{O}(h s_i) = \mathrm{o}(h)$ as $d \rightarrow \infty$.

Therefore, applying Theorem \ref{thm:ng2} we find
$$
	\mathrm{E}_{\pi_d}[(q_i^T(x_i' - x_i))^2] = 2 h \Phi\left( -\frac{l^3 |\theta-\tsfrac{1}{2}| \sqrt{\tau}}{4} \right) + \mathrm{o}(h)
$$
as $d \rightarrow \infty$, where $\pi_d$ is given in \eqref{eq hat1}.  Reversing the coordinate transformation $\hat{x} = V^{-1/2} x$ we obtain \eqref{eq jump2}.
\end{proof}

\subsection{Proof of Theorem \protect\ref{thm:lstep}}

\begin{proof}
The proof of Theorem \ref{thm:lstep} is similar to the proof of Theorem \ref{thm genlang conv} so we only give brief details here.  Define $t_{d,i} = h \lambda_i^2 = \mathcal{O}(d^{-1/3})$.  For SLA, $G = I - \frac{h}{2}A$, $\Sigma = hI$, $\mathcal{A} = A - \frac{h^2}{2} A^2$ and $\mathcal{A}^{-1}\beta = A^{-1}b$.  Multi-step SLA has the same $\mathcal{A}$ and $\mathcal{A}^{-1}\beta$, so $\tilde{\lambda}_i^2 = (1-\frac{1}{4} t_{d,i})\lambda_i^2$, $\hat{r}_i = 0$, $r_i = \frac{1}{4}t_{d,i}$, and $\tilde{r}_i = (1-\frac{1}{4}t_{d,i})^{-1}$.

The eigenvalues of $G_L = G^L$ are $G_i^L = (1-\frac{1}{2} t_{d,i})^L = 1 - \frac{L}{2} t_{d,i} + \mathcal{O}(t_{d,i}^2)$, so $\tilde{g}_i = \frac{L}{2} t_{d,i} + \mathcal{O}(t_{d,i}^2)$ and $\hat{g}_i = L t_{d,i} + \mathcal{O}(t_{d,i}^2)$.  Hence $\tilde{g}_i^2 \hat{r}_i^2 \lambda_i^2 = 0$, $\tilde{g}_i = \mathcal{O}(t_{d,i}^2)$ and $\hat{g}_i = \mathcal{O}(t_{d,i})$.

$T_{0i} = T_{1i} = T_{2i}=0$, $T_{3i} = \frac{L}{8} t_{d,i}^2 + \mathcal{O}(t_{d,i}^3)$, $T_{4i} = -\frac{L}{8} t_{d,i}^2 + \mathcal{O}(t_{d,i}^3)$, and $T_{5i} = - \frac{1}{4} t_{d,i} (L t_{d,i} + \mathcal{O}(t_{d,i}^2))^{1/2} (1 + \mathcal{O}(t_{d,i}))$, so $T_{3i}$ and $T_{4i}$ are $\mathcal{O}(d^{-2/3})$ while $T_{5i} = \mathcal{O}(d^{-1/2})$.

This leads to $\lim_{d\rightarrow \infty} \sum_{i=1}^d T_{3i}^2 = 0$ and $\lim_{d\rightarrow \infty} \sum_{i=1}^d T_{4i}^2 = 0$, so $\sigma^2 = \lim_{d\rightarrow\infty} \sum_{i=1}^d T_{5i}^2 = \frac{l^6 L \tau}{16}$.  Also, $\mu = - \frac{l^6 L \tau}{32}$.

Hence we obtain \eqref{eq:lstepa}, using Theorem \ref{thm:ng1} and Corollary \ref{cor:ng}.  

For the expected jump size use Theorem \ref{thm:ng2} with $\hat{r}_i=0$, $\hat{g}_i \tilde{\lambda}_i^{-2} = Lh + \mathrm{o}(h)$, $\tilde{g}_i^2 \lambda_i^{-2} = \mathrm{o}(h)$ and $\gamma_{d,i}$ bounded uniformly.
\end{proof}

\subsection{Proof of Theorem \protect\ref{thm uhmc conv}}

\begin{proof}
Define a spectral decomposition $B = Q \Lambda Q^T$ where $Q$ is an orthogonal matrix and $\Lambda = \operatorname{diag}(\lambda_1^2,\dotsc,\lambda_d^2)$ is a diagonal matrix of eigenvalues of $B$.  Define
$$
	\tilde{Q} = 	\left[ \twobytwo{Q}{0}{0}{Q} \right] 
	\quad \mbox{and} \quad
	\tilde{K} = \left[ \twobytwo{I - \tsfrac{h^2}{2}\Lambda}{hI}{-h\Lambda + \tsfrac{h^3}{4} \Lambda^2}{I - \tsfrac{h^2}{2}\Lambda} \right].
$$
Then $\mathcal{K} = \tilde{Q} \tilde{K} \tilde{Q}^T$, $\mathcal{K}^L = \tilde{Q} \tilde{K}^L \tilde{Q}^T$ and $(\mathcal{K}^L)_{11} = Q (\tilde{K}^L)_{11} Q^T$.  Thus $(K^L)_{11}$ and $(\tilde{K}^L)_{11}$ have the same eigenvalues.

Notice that $\tilde{K}$ is a $2\times2$ block matrix where each $d\times d$ block is diagonal.  Therefore, $\tilde{K}^L$ is also a $2\times2$ block matrix with diagonal blocks.  In particular, $(\tilde{K}^L)_{11}$ is a diagonal matrix, with eigenvalues on the diagonal.  Moreover, 
$$
	[(\tilde{K}^L)_{11}]_{ii} = (k_i^L)_{11}
$$
where $[(\tilde{K}^L)_{11}]_{ii}$ is the $i^{\mathrm{th}}$ diagonal entry of $(\tilde{K}^L)_{11}$,  $(k_i^L)_{11}$ is the $(1,1)$ entry of the matrix $k_i^L \in \bbR^{2\times2}$, and $k_i \in \bbR^{2\times2}$ is defined by
$$
	k_i = \left[ \twobytwo{(\tilde{K}_{11})_{ii}}{(\tilde{K}_{12})_{ii}}{(\tilde{K}_{21})_{ii}}{(\tilde{K}_{22})_{ii}} \right]
	= \left[ \twobytwo{1-\tsfrac{h^2}{2}\lambda_i^2}{h}{-h\lambda_i^2 + \tsfrac{h^3}{4}\lambda_i^4}{1-\tsfrac{h^2}{2}\lambda_i^2} \right].
$$  
This matrix can be factorized
$$
	k_i = \left[ \twobytwo{1}{0}{0}{a} \right] \left[ \twobytwo{\cos(\theta_i)}{-\sin(\theta_i)}{\sin(\theta_i)}{\cos(\theta_i)} \right]
	\left[ \twobytwo{1}{0}{0}{a^{-1}} \right]
$$
where $a = \lambda_i \sqrt{1 - \tsfrac{h^2}{4}\lambda_i^2}$ and $\theta_i = -\cos( 1-\tsfrac{h^2}{2}\lambda_i^2)$.  Therefore,
$$
	k_i^L = \left[ \twobytwo{1}{0}{0}{a} \right] \left[ \twobytwo{\cos(L \theta_i)}{-\sin(L \theta_i)}{\sin(L \theta_i)}{\cos(L \theta_i)} \right]
	\left[ \twobytwo{1}{0}{0}{a^{-1}} \right],
$$
and $[(\tilde{K}^L)_{11}]_{ii} = (k_i^L)_{11} = \cos(L\theta_i)$.
\end{proof}

\subsection{Proof of Theorem \protect\ref{thm hmc conv}}

\begin{proof}
By Theorem \ref{thm iso hmc}, instead of analyzing HMC with proposal \eqref{hmc prop} and target $\pi_d$, we can analyze MH algorithm \eqref{eq hmchat1}.  By Theorem \ref{thm uhmc conv}, the eigenvalues of $G$ for the proposal of \eqref{eq hmchat1} are $G_i = \cos(L \theta_i)$ where $\theta_i = -\cos^{-1}(1-\frac{h^2}{2} \lambda_i^2)$, so 
$$
	\tilde{g}_i = 1-\cos(L \theta_i) \qquad \mbox{and} \qquad \hat{g}_i = \sin^2(L \theta_i).
$$
Similarly, since $\mathcal{A} = (\mathcal{K}^L)_{12}^{-2} (I - (\mathcal{K}^L)_{11}^2)$, using the same notation as in the proof of Theorem \ref{thm uhmc conv}, 
$$
	\tilde{\lambda}_i^2 = (k_i^L)_{12}^{-2} (1 - (k_i^L)_{11}^2) = \frac{ 1- \cos^2(L \theta_i) }{(a_i^{-1} \sin^2(L \theta_i))^{2}} = a_i^2 = \lambda_i^2 \left( 1 - \tsfrac{1}{4} s_i \right)
$$
where $a_i = \lambda_i \sqrt{1 - \tsfrac{h^2}{4}\lambda_i^2}$ and $s_i = h^2 \lambda_i^2 = l^2 d^{-1/2} (\frac{\lambda_i}{d^\kappa})^2 = \mathcal{O}(d^{-1/2})$.  Hence,
\begin{align*}
	r_i &= \tsfrac{1}{4} s_i, & \tilde{r}_i &= \tsfrac{1}{1-\tsfrac{1}{4}s_i}, & \hat{r}_i &= 0.
\end{align*}
Note that $\mathcal{A}^{-1}\beta = A^{-1}b$ for the untransformed HMC algorithm, and $\hat{r}_i=0$ is preserved by the coordinate transformation.  Hence, $T_{0i} = T_{1i} = T_{2i} = 0$, and
\begin{align*}
	T_{3i} &= \tsfrac{1}{8} s_i \sin^2(L\theta_i), &
	T_{4i} &= - \frac{\tsfrac{1}{8} s_i \sin^2(L\theta_i)}{1-\tsfrac{1}{4}s_i}, &
	T_{5i} &= - \frac{ \tsfrac{1}{8} s_i \sin(2L\theta_i)}{\sqrt{1-\tsfrac{1}{4}s_i}}.
\end{align*}
Using the trigonmetric expansion $\cos^{-1}(1-z) = \sqrt{2z} + \mathcal{O}(z^{3/2})$, and defining $T'$ such that $L = \tsfrac{T'}{h}$ we find 
$$
	L \theta_i
			= L ( - \sqrt{s_i} + \mathcal{O}(s_i^{3/2}) )
			= - L s_i^{1/2} (  1 + \mathcal{O}(s_i) )
			= - T' \lambda_i ( 1 + \mathcal{O}(d^{-1/2}) )
$$
hence, there exists a function $T''(d)$ such that $L \theta_i = - T'' \lambda_i$ and
$T''(d) = T' + \mathcal{O}(d^{-1/2})$.

To apply Theorems \ref{thm accept} and \ref{thm jumpsize} we need to check \eqref{eq Tcond}.  For some $h > 0$, $c l^2 (\tsfrac{i}{d})^{2\kappa} \leq d^{1/2} s_i = l^2 (\tsfrac{\lambda_i}{d^\kappa})^2 \leq C l^2 (\tsfrac{i}{d})^{2\kappa}$ and so we find
\begin{multline*}
	\lim_{d \rightarrow \infty} \frac{ \sum_{i=1}^d |T_{3i}|^{2+\delta} }{ \left( \sum_{i=1}^d |T_{3i}|^2 \right)^{1+\delta/2} }
	= \lim_{d \rightarrow \infty} \frac{ \sum_{i=1}^d |s_i \sin^2(L\theta_i) |^{2+\delta} }{ \left( \sum_{i=1}^d |s_i \sin^2(L\theta_i)|^2 \right)^{1+\delta/2} } \\
	= \lim_{d \rightarrow \infty} d^{-\delta/2} \left( \frac{ \left( \sum_{i=1}^d | d^{1/2} s_i \sin^2(T'' \lambda_i) |^{2+\delta} \right)^{1/(2+\delta)} }{ \left( \sum_{i=1}^d |d^{1/2} s_i \sin^2(T'' \lambda_i)|^2 \right)^{1/2} } \right)^{2+\delta}
	= 0
\end{multline*}
since $\nrm{v}_{2+\delta} \leq \nrm{v}_2$ for all $v \in \bbR^d$.
Similar arguments verify \eqref{eq Tcond} for $T_{4i}$ and $T_{5i}$.  With
\begin{align*}
	\mu 
	&= \lim_{d \rightarrow \infty} \sum_{i=1}^d T_{3i} + T_{4i} 
	= -\frac{l^4 \tau}{32}
\end{align*}
and 
\begin{align*}
	\sigma^2 &= \lim_{d \rightarrow \infty} \sum_{i=1}^d 2 T_{3i}^2 + 2 T_{4i}^2 + T_{5i}^2 
	= \frac{l^4 \tau}{16},
\end{align*}
we obtain $\frac{\mu}{\sigma} = -\sigma - \frac{\mu}{\sigma} = - \frac{l^2 \sqrt{\tau}}{8}$ and $\mu+\sigma^2/2 = 0$, so from Theorem \ref{thm accept} we obtain \eqref{eq hmcexpect}.

For \eqref{eq hmcjump}, we apply Lemma \ref{lem jump} since $\mu_{d,i} = \mathcal{O}(d^{-1})$ and $\sigma_{d,i}^2 = \mathcal{O}(d^{-1})$, using
$$
	\left( \frac{1+G_i}{\tilde{\lambda}_i^2} + \frac{1-G_i}{\lambda_i^2} \right) = \frac{2}{\lambda_i^2} + \mathcal{O}(h^2), \quad \mbox{as } d \rightarrow \infty.
$$
\end{proof}

\bibliographystyle{imsart-number}
\bibliography{paper8bib} 

\end{document}